\documentclass[11pt]{article}
\usepackage[utf8]{inputenc}
\usepackage{xspace}
\usepackage{setspace}
\usepackage[margin=1in,a4paper]{geometry}
\usepackage{graphicx}
\graphicspath{ {./figures/} }
\usepackage{subcaption}
\usepackage{amsmath,amsfonts,amssymb}
\usepackage{xcolor}
\usepackage[colorlinks=true,
            citecolor=blue]{hyperref}
\usepackage{fullpage}
\usepackage{amsthm}
\usepackage{algorithm}
\usepackage{algpseudocode}
\usepackage{bbm}
\usepackage{mathtools}
\usepackage{float}
\usepackage{indentfirst}
\usepackage{tikz}
\usepackage{accents}
\usetikzlibrary{arrows.meta}
\usetikzlibrary{shapes.misc}
\tikzset{cross/.style={cross out, draw=black, minimum size=2*(#1-\pgflinewidth), inner sep=0pt, outer sep=0pt},
cross/.default={1pt}}
\allowdisplaybreaks
\newtheorem{theorem}{Theorem}[section]
\newtheorem*{theorem*}{Theorem}
\newtheorem{corollary}{Corollary}[section]
\newtheorem{lemma}{Lemma}[section]

\newtheorem*{prop*}{Proposition}

\theoremstyle{definition}
\newtheorem{definition}{Definition}[section]
\newtheorem*{definition*}{Definition}

\theoremstyle{remark}

\numberwithin{equation}{section}
\usepackage[capitalize]{cleveref}
\def\submission{0}

\newcommand{\poly}{\operatorname{poly}}
\newcommand{\sgn}{\operatorname{sgn}}

\newcommand{\bin}{\operatorname{Bin}}
\newcommand{\fdim}{\widehat{\dim}}
\newcommand{\supp}{\operatorname{supp}}
\newcommand{\wt}{\operatorname{wt}}
\newcommand{\var}{\operatorname{Var}}
\newcommand{\yes}{\textsf{yes}\xspace}
\newcommand{\no}{\textsf{no}\xspace}

\newcommand{\vspan}{\operatorname{span}}

\newcommand{\dist}{\operatorname{dist}}
\newcommand{\etal}{\emph{et al.\@}}

\usepackage[style=alphabetic,
sorting=none]{biblatex}
\providecommand{\email}[1]{\href{mailto:#1}{\nolinkurl{#1}\xspace}}

\title{Exponential Quantum Advantage in Testing Fourier Dimensionality}

\ifnum\submission=0 % Alphabetical order
\author{Kenny Chen\thanks{The University of Sydney, Email: \email{kche5493@uni.sydney.edu.au}}}
\else 
\author{Author(s) redacted}
\fi

\date{September 2026}

\begin{document}
\maketitle
\begin{abstract}
    A boolean function $f$ has Fourier dimension $k$ if its strings with nonzero Fourier coefficients span a subspace of dimension $k$. We first consider the property testing task of determining whether a function has Fourier dimension at most $k$, or is $\epsilon$-far from being so. We show that there is an $O(k/\sqrt{\epsilon})$-query quantum property tester for this problem, which we show to be almost optimal. Combined with Gopalan \etal's classical lower bound of $\Omega(2^{k/2})$, this demonstrates an exponential quantum advantage for this task \cite{DBLP:journals/siamcomp/GopalanOSSW11}. We complement this result with a $\tilde{O}(2^{k/2}/\epsilon)$ classical tester, giving a quadratic improvement over the previous best tester, and essentially settling the classical query complexity.

    We then introduce and study the tolerant version of this problem. That is, determining whether a function is $\alpha\epsilon$-close to having Fourier dimension at most $k$, or is $\epsilon$-far from being so. Perhaps surprisingly, we show that this tolerant version of the problem is essentially just as hard as the standard version of the problem. That is, we show this problem has a $\tilde{O}(k/\epsilon)$ quantum tester, and a $\tilde{O}(2^{k/2}/\epsilon)$ classical tester.
\end{abstract}
\section{Introduction}
One of the key problems in the field of quantum query complexity is trying to quantify the advantage that quantum computers have over classical computers. A key paradigm in which this has been studied is in the \emph{query complexity model}, where an algorithm is given black box access to a function(s), and is tasked to determine some property of the function. Here efficiency of the algorithm is measured in the number of queries needed to this black box. One area which naturally lends itself to the study of advantage in this model is that of \emph{property testing} of boolean functions, since the natural classical measure is already that of query complexity. Here, numerous properties have been studied, such as that of testing for $k$-juntas \cite{DBLP:journals/qip/AticiS07}, monotonicity \cite{DBLP:journals/toc/BelovsB15} and linearity \cite{chakraborty2013improvedquantumtestlinearity}, to name a few. However, perhaps more interestingly than whether such properties can be tested more efficiently is the question of \emph{how} much more efficiently these properties can be tested. That is, what is the largest gap in query complexity one can create between quantum and classical computers? 

This question began outside of property testing. One of the first demonstrations of exponential advantage comes from Simon's algorithm, which can be rephrased as distinguishing between whether a function has the property of being one-to-one or periodic (with the promise that the function satisfies one of the two). Here it was shown that a quantum computer could distinguish between these two with $O(n)$ quantum queries, whereas a classical computer required $\Omega(2^{n/2})$ queries to do so \cite{DBLP:journals/siamcomp/Simon97a}. This gap was later pushed by the flagship Forrelation problem, where Aaronson and Ambainis essentially showed there was a property which was testable with $O(1)$ quantum queries, yet required $\Omega(2^{n/2})$ classical queries \cite{DBLP:journals/siamcomp/AaronsonA18}. 

While these demonstrably show advantage for quantum computers, and can (potentially) be interpreted under a property testing lens, there are (at least) two reasons why property testers may find these results unsatisfying. First is that, whenever possible, one is interested in properties that are efficiently testable (where ``efficiently'' here means independent of $n$). While these above problems do exhibit exponential advantage, they are not efficiently testable classical properties. Addressing this concern, \cite{doi:10.1137/23M1573975} did succeed at exhibiting a graph property which was testable with $O(k)$ quantum queries, yet required $\Omega(2^{k/2})$ classical queries to test. In this sense, it is a classical property that is testable independent of $n$, resolving this first issue. This leads to the second issue, which is that whilst these are very important and interesting theoretical results, one could object that they are not ``natural'' problems, in the sense that they were created for the specific purpose of exhibiting quantum advantage and do not, \emph{per se}, correspond to ``real-world'' applicable problems. As such, in light of these two objections, one could instead turn to the following question:
\begin{center}
    \textit{Is there exponential quantum advantage for a class of more natural, efficiently classically testable properties?}
\end{center}
Very recent work has shown that this might exist in the problem of tolerant junta testing. That is, for certain parameter regimes, Tal and Yuan showed the task of distinguishing whether a function $f$ is $\epsilon_1$ close to some $k$-junta, or $\epsilon_2$ far from every $k$-junta can be solved with $\poly(k)$ quantum queries, but required at least $k^{\Omega(\log k)}$ classical queries, the first super-polynomial quantum advantage for adaptive tolerant junta testing \cite{DBLP:conf/coco/TalY26}. However, this is not quite the exponential separation as achieved by that of Simon's problem and Forrelation.\footnote{Indeed, Tal and Yuan leave as an open question whether tolerant junta testing is one of these properties that admit an exponential advantage.} If we allow ourselves to compare to passive property testing (where the classical queries are obtained at random, instead of an algorithm being able to choose where the queries go), then Caro, Naik and Slote do show the problems of symmetry and triangle freeness do exhibit this exponential advantage \cite{caro_et_al:LIPIcs.ITCS.2026.34}. However, this advantage disappears once you allow the classical algorithm to choose its queries, as opposed to getting random samples, and this former setting is the one we are interested in. Thus we can continue looking for more natural problems which might give us such an exponential quantum advantage.

One natural property to study, with the above motivation in mind, is that of Fourier dimension. At a high level, a boolean function $f$ has Fourier dimension $k$ if the span of its strings with nonzero Fourier coefficients is a subspace of dimension $k$ (see \cref{def:fourier_dim}). The task then, is to decide whether $f$ has Fourier dimension at most $k$, or if it is far from having Fourier dimension $k$. This task we call Fourier dimensionality testing (see \cref{def:fourier_dim_testing}). Though it may not immediately appear so, the Fourier dimension of a function is an interesting property to study in its own right. Looking at broader complexity theory, the Fourier dimension of a function is equal to its non adaptive parity decision tree depth \cite{DBLP:journals/toc/Sanyal19}. This means if it has Fourier dimension $k$, it requires $k$ XOR queries (chosen all at once nonadaptively) to uniquely determine its output on any input. Considering just XOR functions, Montanaro and Osborne also showed the Fourier dimension of an XOR function is equal to its deterministic one way communication complexity, both classically and quantumly \cite{DBLP:journals/corr/abs-0909-3392}. This also has many implications, for example to sketching, since every XOR function is a linear sketch over $\mathbb{F}_2$ \cite{DBLP:conf/coco/KannanMSY18}. Turning specifically towards the property testing motivations, Fourier dimensionality can be thought of as aa generalization of junta testing. That is, functions that have Fourier dimension $k$ are exactly linear $k$-juntas.\footnote{A function $f:\mathbb{F}_2^n\rightarrow\{-1,1\}$ is a linear $k$-junta if there are $k$ linearly independent strings $u_1,\dots,u_k$ and a function $g:\mathbb{F}_2^k\rightarrow\{-1,1\}$ such that $f(x)=g(\langle u_1,x\rangle,\dots,\langle u_k,x\rangle)$.} One can think of this as testing for $k$-juntas, but in a basis-independent manner, as opposed to regular $k$-junta testing, which is a linear $k$-junta in the canonical basis. Indeed, this is the problem studied by De, Mossel and Neeman, albeit they do so for real input functions \cite{DBLP:conf/colt/DeMN19}. As \cite{DBLP:conf/colt/DeMN19} points out, many natural classes of functions are included here, including $k$-juntas themselves, and functions of halfspaces. Finally, \cite{DBLP:journals/siamcomp/GopalanOSSW11} also observe that having a small Fourier dimension implies having a small hidden truth table, meaning one can actually implicitly learn the truth table with a $\poly(2^k,1/\epsilon)$ tester. 

\cite{DBLP:journals/siamcomp/GopalanOSSW11} first studied this problem of testing Fourier dimensionality, and gave an $O(k2^{2k}/\epsilon)$ tester \cite[Theorem 5.3]{DBLP:journals/siamcomp/GopalanOSSW11}. They also complemented this result with an $\Omega(2^{k/2})$ lower bound for any randomized algorithm \cite[Theorem 7.1]{DBLP:journals/siamcomp/GopalanOSSW11}. The upper bound was then improved by Alekseychuk and Konyushok, who gave an upper bound of $O(k^22^k/\epsilon)$ \cite{10.1007/s10559-013-9498-z}. 

There are a couple of reasons for why Fourier dimensionality is a good candidate for exhibiting a ``real-world'' exponential advantage. For one, as surveyed above, it is a previously studied problem with independent interest. But specifically for property testing, it is a problem with a known exponential (in $k$) lower bound. Furthermore, one of the main advantages quantum computers have over classical computers is their ability to easily access the Fourier transform via the Hadamard transform. Thus, whilst the main difficulty classically may be accessing the Fourier transform, there is hope that quantumly, one can get around this.  This leads us to the second, more specific question, which is the main focus of this paper:
\begin{center}
    \textit{Does the problem of Fourier dimensionality testing exhibit exponential query advantage?}
\end{center}
\subsection{Our Results}
\subsubsection{Standard Testing}
Our first set of results regards the (standard) Fourier dimensionality testing problem. We begin by resolving the above question positively. Namely, there is a quantum tester for Fourier dimensionality with the following query complexity.
\begin{theorem}[Informal, see \cref{thm:quantum_tester}]\label{inf:quantum_ub}
    There is a quantum algorithm using $O(k/\sqrt{\epsilon})$ queries that distinguishes between a function having Fourier dimension $k$, and being $\epsilon$-far from having Fourier dimension at most $k$ with probability at least $2/3$.
\end{theorem}
We complement this result with a lower bound tight in $k$, showing our tester is essentially optimal:
\begin{theorem}[Informal, see \cref{thm:quantum_lower_bound}]Any quantum tester for Fourier dimensionality requires at least $\Omega(k)$ queries.
\end{theorem}

Our second contribution to this topic is to improve the classical tester for Fourier dimensionality. We show the following:
\begin{theorem}[Informal, see \cref{thm:classical_tester}]
    There is a classical tester which uses $\tilde{O}(2^{k/2}/\epsilon)$ queries, and distinguishes between a function having Fourier dimension at most $k$, and $f$ being $\epsilon$-far from having Fourier dimension $k$ with probability at least $2/3$.
\end{theorem}
This has two main consequences. First, it is a tester that is essentially tight against the lower bound of $\Omega(2^{k/2})$ presented by \cite{DBLP:journals/siamcomp/GopalanOSSW11}, up to lower order terms \cite[Theorem 7.1]{DBLP:journals/siamcomp/GopalanOSSW11}. Second, combined with the above, this not only shows that testing Fourier dimensionality, and thus as an extension, testing for $k$-linear juntas, is a problem which exhibits exponential quantum advantage, but also exactly quantifies the extent of this advantage. That is, Fourier dimensionality testing exhibits the type of advantage exhibited by Simon's problem, though arguably for a much more natural and less contrived problem ($\Theta(k)$ quantum queries against $\Theta(2^{k/2})$ classical queries).
\subsubsection{Tolerant Testing}
Secondly, following from the work of Tal and Yuan, we initiate and study the problem of \emph{tolerant} Fourier dimensionality testing. Broadly, tolerant testing was first studied by Parnas, Ron and Rubinfeld, when they noted that the usual property testing condition of accepting \emph{only} when some function exactly fulfills some property was too restrictive, and unrealistic in real world applications \cite{DBLP:journals/jcss/ParnasRR06}. Hence, they studied tolerant testing where now the aim is to accept when some function is close to exhibiting some property, and reject when it is far from doing so. We introduce this problem to Fourier dimensionality testing (see \cref{def:tolerant_fdim_testing} for a formal definition). Though we discuss this more in detail in the technical sections, we will just note that it is not in general true that a regular property tester directly leads to a tolerant property tester. We refer to the survey by Ron for a discussion on when tolerant testing is easier, or just as hard as standard testing \cite{DBLP:conf/icalp/Ron25}, but usually tolerant testing is substantively more difficult and requires new ideas. Thus, as a high level summary of the following results, we show that, perhaps surprisingly, tolerant testing of Fourier dimensionality is just as hard as standard testing. As a consequence, for example, though tolerant testing of juntas is much more difficult than standard testing, tolerant testing of \emph{linear} juntas is no more difficult than standard testing.

One might also ask why we are interested in the prior results on standard testing, if we have results on tolerant testing, that show it is just as hard as standard testing. Indeed, the tolerant results are more general, and encompass the standard testing results. There are two main reasons. The first is that, whilst the tolerant testing results are inspired by results used in the standard setting, they are still different algorithms. One consequence of this is that the testers in the standard problem have better dependence on the lower order terms than the tolerant tester (for example, better dependence on $\epsilon$, no $\log$ factors). They also do use different underlying ideas, even though they may use some of the same lemmas and observations. Secondly, as we will discuss more in \cref{ss:techniques}, though the tolerant testers are essentially query optimal, they are by no means efficient algorithms. That is, the tolerant tester requires brute force checking all $k$-dimensional subspaces, of which there are $O(2^{O(nk)})$ of them, whereas our standard tester does not. This means the query complexity is more reflective of the running time of the standard testers, than that of the tolerant testers. While of course not an issue for query bounds, this is a strong enough caveat that we feel the standard and tolerant results should be stated separately, instead of the former being viewed as a special case of the latter.

We start this section by exhibiting a tolerant quantum tester:
\begin{theorem}[Informal, see \cref{thm:tolerant_quantum_ub}]\label{inf:tolerant_quantum_ub}
    There is a quantum tester for tolerant Fourier dimensionality that uses $\tilde{O}\left(\frac{k}{\epsilon}\right)$ quantum queries.
\end{theorem}
This shows that, at least in terms of $k$, tolerant testing is no harder than regular testing. As we discuss more below, despite the fact that a regular tester does not in general lead to a tolerant tester, we adapt and extend the techniques in the standard quantum tester, which allows us to obtain an algorithm for the tolerant version of the problem. We expand on this more in \cref{sec:tolerant_quantum_tester}, but we are unable to recover the exact upper bound in \cref{inf:quantum_ub} using our techniques. This is because we are unable to use amplitude amplification in the tolerant algorithm to get a square root dependence on $\epsilon$. We leave it as an open problem as to whether this bound can be improved. 

Finally, we also give a classical tolerant tester for this problem.
\begin{theorem}[Informal, see \cref{thm:classical_tolerant_ub}]
    There is a classical tester for tolerant Fourier dimensionality testing that uses $\tilde{O}\left(\frac{2^{k/2}}{\epsilon}\right)$ classical queries. 
\end{theorem}
Despite this bound looking the same as the bound in the standard version of the problem, as we will discuss in the next section, we need to develop much more novel machinery to prove this bound. Most key to the analysis is a new result regarding probability compressions which may be of independent interest. Though not a generalization of their methods, this can be thought of as inspired by the techniques of Alekseychuk and Konyushok \cite{10.1007/s10559-013-9498-z}.
\begin{theorem}[Informal, see \cref{thm:probability_compression_theorem}]
    Let $\mu$ be any probability distribution on $\mathbb{F}_2^n$, $X:\mathbb{F}_2^n\rightarrow\mathbb{F}_2^r$ be a uniformly random linear map with $r=\tilde{O}(k)$, and let $\rho_k(\mu)$ be the minimum probability mass outside any subspace of dimension $k$. Then, for any constant $0<\beta<1$, with probability at least $0.99$,
    \begin{equation*}
        \rho_k(\mu)-\beta\leq \rho_k(X\mu)\leq \rho_k(\mu).
    \end{equation*}
\end{theorem}
At a high level, this states that we can apply a random linear map to ``compress'' the domain, and the approximate minimum mass outside any subspace of constant dimension $k$ is preserved.

\subsection{Our Techniques}\label{ss:techniques}
\subsubsection{Quantum Fourier Dimensionality Tester}
 We outline the quantum Fourier dimensionality tester. The main idea is as follows: A quantum algorithm can sample from the Fourier spectrum of the function $f$. More concretely, in one quantum query, we can sample a string $s$, with probability equal to $\hat{f}(s)^2$, where $\hat{f}(s)$ is the Fourier coefficient of the string $s$. Then, to distinguish $f$ from having Fourier dimension $k$, or being far from having Fourier dimension $k$, at a high level it suffices to continue taking samples until we have a subspace of dimension $k+1$, or until with high probability, we are confident this is not the case. There are three things of note to make this formal. First, clearly if the function has Fourier dimension at most $k$, this algorithm will accept. Next, this only works if at every step, we have a nontrivial probability of sampling a string outside the current subspace induced by the previously sampled vectors. Indeed, it may be the case that although $f$ is far from having Fourier dimension $k$, we cannot actually sample outside this subspace with any good probability. We show this is not the case, and that if $f$ is far from having Fourier dimension $k$, then it always has $\Omega(\epsilon)$ Fourier mass outside of any subspace of dimension $k$ (see \cref{lem:Fourier_mass} and \cref{cor:fourier_mass} for the formal statements). Thus, at every time point, we can use Fourier sampling and find with probability $\Omega(\epsilon)$ a string that extends the current subspace, until our current subspace has dimension $k+1$, upon which point we can reject. The third technicality, however, is doing the above would only yield a tester that has $O(1/\epsilon)$ dependence, since we need this many samples to see a string that extends the subspace. To improve this dependence on $\epsilon$, instead of naively using Fourier sampling, we first use amplitude amplification to amplify the magnitude on strings outside the current subspace. Doing this uses $O(1/\sqrt{\epsilon})$ queries, but ensures that after measurement we, with high probability, obtain a string that does actually extend the subspace, giving the improved dependence on $\epsilon$.

It is worth commenting that although the tester above seems simple in hindsight, the second crucial observation: that $f$ being far from having Fourier dimension $k$ implies it always has nontrivial mass outside of any subspace of dimension $k$, while convenient, is not necessarily obvious or trivial. In general, a property of the Fourier spectrum cannot always be detected via Fourier sampling. For example, one can consider the seemingly related task of testing whether a function has Fourier degree $k$, or is $\epsilon$-far from having Fourier degree $k$.\footnote{As a high level summary, a function has Fourier degree $k$ if the Fourier coefficient associated with all strings that have Hamming weight at least $k+1$ are zero.} One could try to design a similar quantum tester for this problem: that is, repeatedly use Fourier sampling (with or without amplitude amplification), and reject upon finding a string with high hamming weight. The assumption required to make this work is that if $f$ is $\epsilon$-far from having Fourier degree $k$, then it will have a large enough Fourier mass on strings with Hamming weight at least $k+1$ to be detected efficiently with Fourier sampling. However, this is not the case. Consider the address function, which is a $k+2^k$ junta defined by choosing $k$ ``address'' bits $\{a_i\}_{i\in[k]}$, and $2^k$ addresses $b_i\in\{\pm1\}$. Then, the address function is $f(a,b)=b_a$, which can be padded out to an $n$-variate function. Then, this function can be constructed to be constant far from having Fourier degree $k$, but have total Fourier mass on levels greater than $k$ equal to $O(2^{-k})$. Hence, it will take $\Omega(2^k)$ Fourier samples to find one of these strings, meaning we do not get a linear quantum tester. The key issue is that in this case, being far from Fourier degree $k$ does not imply there is non-negligible Fourier mass on the higher levels, unlike in the Fourier dimensionality case, where being far from Fourier dimension $k$ does imply non-negligible Fourier mass outside any subspace of dimension $k$. Thus, the main technical observation is showing that the ``simple intuition'' actually works.

To prove this corresponding lower bound, we reduce to the decision Simon's problem, which has an $\Omega(k)$ lower bound, as shown by Koiran, Nesme and Portier \cite{DBLP:conf/icalp/KoiranNP05}. We show that functions on $\mathbb{F}_2^{k+1}$ that are periodic have Fourier dimension at most $k$, whereas functions that are one to one are far, and have Fourier dimension $k+1$. This can then be padded to a function on $\mathbb{F}_2^n$ by padding to redundant variables. While Simon's problem concerns functions with outputs in $\{0,1\}^n$, whereas our function is boolean valued, we show one can apply a uniform random boolean function $M:\{0,1\}^n\rightarrow\{0,1\}$ as a mask, which allows us to use both the Simon's lower bound and our Fourier dimensionality tester.
\subsubsection{Classical Fourier Dimensionality Testing}
The design of our classical tester is inspired by that of \cite{DBLP:journals/siamcomp/GopalanOSSW11}, and thus it helps to overview their approach first. For the sake of exposition, assume the strings with nonzero Fourier coefficients in $f$ span the subspace $U$. Then, the Fourier dimension of $f$ equals the dimension of $U$. The main observation is that $f$ is periodic on strings in $U^\perp$. Then, they sample $\ell=O(2^k)$ points $x_i$, and $m=O(k2^k/\epsilon)$ shifts $h_j$ for each $x_i$. Then, if it is the case that $f(h_j)=f(h_j+x_i)$ for all $h_j$, then $x_i$ likely is in $U^\perp$, and they add it to a tracking set $H$. They then show that the size of $H$ is different for if $f$ has Fourier dimension $k$, as opposed to when $f$ is far from having Fourier dimension $k$, and test for this difference in size. Intuitively, one can see that if $f$ has Fourier dimension at most  $k$, then $|U^\perp|\geq2^{n-k}$, whereas if $f$ is far from having Fourier dimension at most $k$, then $|U^\perp|\leq 2^{n-k-1}$, and one can try to test for this gap.

Our starting point uses the same observation, but in a different manner. We start by observing that two strings $x$ and $y$ have $f(x)=f(y)$ if they are in the same coset induced by $U^\perp$ (see \cref{ss:group} if one is unfamiliar with a coset). This also means for any shift $z$, $f(x+z)=f(y+z)$, since after a shift, if two vectors start in the same coset, they will also be shifted into the same coset (though of course, this is not necessarily the one they started in). Further, we note $U^\perp$ induces at most $2^{\dim{U}}$ many cosets. Our classical tester builds on this fact, and a simple balls and bins observation. 

Imagine first, in a very simplified world, we could query a point, and find out which coset it lies in (and that we have some canonical way of identifying the cosets such that this makes sense and does not give anything away to the algorithm). Then, if $f$ has Fourier dimension at most $k$, we know there are at most $2^{k}$ cosets, and thus, by a balls and bins argument, we should expect to see a collision with $O(2^{k/2})$ queries. On the other hand, if $f$ is far from having Fourier dimension $k$, then it has at least $2^{k+1}$ cosets, and thus we should not expect to see (as many) collisions. Thus, we can sample $\ell=O(2^{k/2})$ points, $x_1,\dots,x_\ell$, and test the number of collisions we see, to distinguish between these two cases. 

Of course, the issue with the above approach is that we do not have some canonical way of identifying the cosets, because $f$ is a binary function, and thus is not sufficiently granular. To do this, we additionally sample $m=O(k/\epsilon)$ shifts $h_1,\dots,h_m$, and construct the tuple, for each $x_i$ in the $\ell$ query points,
\begin{equation*}
    \Phi(x_i)=(f(x_i+h_1),f(x_i+h_2),\dots,f(x_i+h_m)).
\end{equation*}
Now, if $x_i$ and $x_j$ ``collide'', in the sense that they are in the same coset, then $\Phi(x_i)=\Phi(x_j)$. On the other hand, if they land in different cosets, then we show we can bound the probability that $\Phi(x_i)=\Phi(x_j)$, and thus bound the probability of seeing collisions. In this way, we can test for collisions in $\Phi(x_i)$, which is actually granular enough to distinguish whether $x_i$ and $x_j$ are in different cosets. Interestingly, the analysis of these classical tuples reconnects to the quantum tester, since we show that bounding the probability that two tuples collide can be reduced to bounding the probability that a sequence of Fourier samples is orthogonal to $x_i+x_j$.
\subsubsection{Quantum Tolerant Fourier Dimensionality Testing}
Recall that in the standard setting, we showed that if a function was $\epsilon$-far from having Fourier dimension at most $k$, then it has at least $\Omega(\epsilon)$ Fourier mass outside of any subspace of dimension $k$. The basis of the tolerant tester is a more refined analysis of the probability mass outside of any dimension $k$ subspace. Namely, we show that if $f$ is close to having Fourier dimension $k$, then it has at most $c_1\epsilon$ Fourier mass outside of some subspace of dimension $k$. Meanwhile, if it is far, it has at least $c_2\epsilon$ Fourier mass outside of any subspace of dimension $k$, where $c_1<c_2$ are constants (dependent on the parameters of the tolerance). Thus, we can take an appropriate number of samples, then check every single $k$-dimensional subspace and count how many queries lie outside each subspace. If there is some subspace with few queries outside, then we know $f$ is close to having Fourier dimension $k$. On the other hand, if every subspace contains many queries outside of it, then we know the function is far from having Fourier dimension $k$. Since the algorithm here is not maintaining a subspace that it slowly grows via Fourier samples, like in the standard testing setting, we will instead rely on brute force checking every $k$-dimensional subspace.
\subsubsection{Classical Tolerant Fourier Dimensionality Testing}
Our results on the classical tolerant tester build on the work of Alekseychuk and Konyushok so it helps to start by giving an overview of their method \cite{10.1007/s10559-013-9498-z}. At a high level, they take a random matrix $X$ of dimension $m\times n$, with $m=\tilde{O}(k)$. Then, they show that if $f$ has Fourier dimension at most $k$, then there is some other function $f'(u)=f(X^\top u)$, which also has dimension at most $k$. The key observation is that $f'(u)$ is an $m$-variable function, instead of an $n$-variable one. Then, one can try to compute the Fourier transform of $f'(u)$ instead, and reject if the dimension is above $k$. 

The reason this is a good starting point, as we explain more in \cref{sec:classical_tolerant}, is that if somehow the algorithm was given a candidate subspace $U$ to test, there is actually a very simple classical algorithm directly implementing the results of Alekseychuk and Konyushok. However, of course, the algorithm is not directly given this candidate subspace $U$, and if the algorithm directly tests all candidate subspaces, there are $2^{O(nk)}$ of these. Unlike for the quantum tester though, where we first sample a set of queries and then test against the subspaces, here we need a separate set of queries for every subspace. Though one can potentially try, with clever recycling of queries between subspaces, to get this number down, it seems highly unlikely that this will lead to any sort of efficient tester. The approach of Alekseychuk and Konyushok then is very promising for this problem, since it allows us to compress our domain onto $O(k)$-length strings, at which point we can actually brute force all the $k$-dimensional subspaces. However, we cannot adopt their result directly. This is because their results make no claim on how the Fourier mass is preserved once the function is compressed down. Specifically for tolerant testing, we need that the minimum Fourier mass outside every $k$-dimensional subspace is ``roughly the same'' before and after compression, so that we can still accurately distinguish between when a function is close to having Fourier dimension at most $k$, and when it is far. As such, inspired by their approach, our first result is to prove a probability compression theorem. This does not rely on any Fourier properties, and thus may be of independent interest. 

Broadly, our result states that for any probability distribution $\mu$, on $\mathbb{F}_2^n$, one can sample a random linear map $X:\mathbb{F}_2^n\rightarrow\mathbb{F}_2^r$, with $r=\tilde{O}(k)$, and observe the mass $\rho_k(\mu)$, which is the minimum probability mass outside of any $k$-dimensional subspace. Then, this quantity $\rho_k$ will roughly remain unchanged before and after applying the random linear map. The proof works in two steps. First, we show with $M=\tilde{O}(k)$ samples from the probability distribution one can, with high probability, obtain an estimate $\tilde{\rho}_k$, which is a good estimate of $\rho_k$. Though seemingly unrelated, this can actually be proven by adapting the ideas in our tolerant quantum tester, and in fact, the proof skeleton and techniques are identical, once adapted for the setting. At this point, we have not yet used the random linear map. The second step is to introduce this random linear map, and show that after applying the random linear map $X$ to the samples, the estimate obtained using the original samples $s_1,\dots,s_M$ and using the samples after applying the linear transformation $Xs_1,\dots, Xs_M$ are, with high probability, identical. 

Applying this theorem, we could obtain a very primitive tester by compressing the function $f$ into domain $\mathbb{F}_2^r$, and then brute forcing against all dimension $k$ subspaces, of which there are now ``only'' $O(2^{kr})\approx O(2^{k^2})$. While this would work, the issue is that this is not as sample efficient as it could be. To remedy this, we bring back the tuple collision idea that we used in the classical tester. For every $k$ dimensional subspace of $\mathbb{F}_2^r$, we define an arbitrary basis $\{w_1,\dots,w_k\}$. Then, we take $\ell$ samples $a_1,\dots,a_\ell$ and for each sample and subspace, we define the tuple
\begin{equation*}
    \Phi_W(a_i)=(\langle w_1,a_i\rangle,\dots,\langle w_k,a_i\rangle).
\end{equation*}
In this case, note $W^\perp$ has dimension $r-k$, and thus we still have $2^{r-(r-k)}=2^k$ cosets induced by $W^\perp$. We show that for every subspace $W$, up to a constant factor, the Fourier mass outside of $W$ is equal to 
\begin{equation*}
    \Pr_{x\in\mathbb{F}_2^n,~d\in W^\perp}[f(x)\neq f(x+X^\top d)].
\end{equation*}
Thus, whenever $\Phi_W(a_i)=\Phi_W(a_j)$, we conclude that $a_i+a_j\in W^\perp$, and it gives us a pair to estimate the above quantity. We can use the same set of $\ell$ queries to perform this estimate for every subspace $W$, and then repeat this until we can estimate the above probability with high probability. This then allows us to determine whether the function is close or far from having Fourier dimension $k$.
\subsection{Related Works}
\paragraph{Related work on linear junta testing. }While not directly studied as testing Fourier dimensionality, De, Mossel and Neeman considered the task of testing linear $k$-juntas \cite{DBLP:conf/colt/DeMN19}, which is equivalent to testing for low Fourier dimension. They specifically work on testing on real valued functions with respect to the Gaussian measure. They show that whilst in general this class is not testable, adding an assumption on the surface area of the boundary of a function allowed for an efficient tester. We also refer to their exposition on how this relates to testing and learning of functions of half-spaces and invariant testing, which are properties also related to testing for linear $k$-juntas, and thus, low Fourier dimension.

\paragraph{Broader property testing advantages. }Quantum advantages have also been studied in other property testing regimes, apart from boolean property testing. For example distributional property testing questions, such as entropy estimation or closeness testing, were studied by Gily\'{e}n and Li \cite{DBLP:conf/innovations/GilyenL20}, though no exponential advantage was proved in this setting. There is also a list of query bounds for distributions regarding problems like hypothesis testing and uniformity testing outlined in \cite{chen2025listcomplexityboundsproperty}, using sample to query lifting, though again, none of exponential advantage. Canonne, Kothari and O'Donnell also show an advantage for uniformity testing \cite{DBLP:conf/tqc/CanonneKO25}. Though this does not exhibit the exponential advantage that is sought after in this paper, it is interesting to note another possible technique for producing advantages. That is, whilst the main tool used in our paper is Fourier sampling, Canonne, Kothari and O'Donnell assume the algorithm also has access to the circuit which produces the samples, and thus allows one to implement the inverse of such a circuit. Such an access model has not only been shown to be more powerful in property testing, but also in quantum PAC learning \cite{salmon24provable}. Ambainis, Childs and Liu also consider the problem of graph property testing, where they show bipartiteness and expansion are properties that can be tested quantumly with polynomial advantage \cite{10.1007/978-3-642-22935-0_31}. Notably, in this model of graph property testing, \cite{doi:10.1137/23M1573975} were able to construct a problem in the adjacency list model that had exponential speedup, although this advantage was not present if the adjacency matrix model was used. Namely, it is a property that has $\poly(k)$ query complexity for a quantum tester, but $\exp(\Omega(k))$ query complexity for a classical tester when tested in the adjacency list model. The property is inspired by the welded trees problem, and whilst it gives the type of property testing exponential advantage we are after, it does have the downside of being a rather contrived property designed to be easy for quantum computers, and not a ``real-world'' motivated problem.

\paragraph{Property testing of quantum objects. }If one expands from just studying advantages, and takes an interest in property testing of quantum properties, this is also a task that is well studied. Here, the goal is to get algorithms that require few copies of the state, rather than few samples from a function. The aim is to be efficient in $d$, the dimension of the state. For example, the problem of testing and learning quantum $k$-juntas, the quantum analogue of $k$-junta testing, was studied by Chen, Nadimpalli and Yuen \cite{DBLP:conf/soda/ChenNY23}, and the tolerant version was studied by Chen, Li and Luo \cite{chen2024tolerantquantumjuntatesting}. O'Donnell and Wright also studied property testing of a state's spectrum (that is, whether a state's spectrum satisfies or is far from satisfying some property) \cite{10.1145/2746539.2746582}. Some other quantum properties studied include complete positivity and separability \cite{buadescu2020lower}, stabilizer state testing \cite{gross2021schur} and matrix product state testing \cite{soleimanifar2022testing}, to name a few.

\paragraph{Tolerant Property Testing. } Tolerant property testing was first introduced by \cite{DBLP:journals/jcss/ParnasRR06}. For broader works, we refer to a talk on tolerant testing by Ron, which surveys many works on classical tolerant testing \cite{DBLP:conf/icalp/Ron25}. Tolerant testing of classical objects has also been studied for quantum algorithms. Of specific relevance for this work is that of \cite{DBLP:conf/sosa/BaoLYYZ26} on non-adaptive quantum testing of juntas, which led to the work of Tal and Yuan showing a superpolynomial quantum advantage in tolerant junta testing \cite{DBLP:conf/coco/TalY26}, the work which inspired our look into tolerant Fourier dimensionality testing. There has also been work on tolerant testing of quantum objects, such as quantum juntas \cite{chen2024tolerantquantumjuntatesting} and stabilizer states \cite{DBLP:conf/stoc/ArunachalamD25}.

\subsection*{Acknowledgements}
We would like to thank Cl\'{e}ment Canonne for pointing us to the work on linear $k$-juntas by De, Mossel and Neeman \cite{DBLP:conf/colt/DeMN19}. We would also like to thank them for discussions on this work. 

\paragraph{LLM Assistance. }LLMs were used in two parts of this paper. First, the more minor one. The author, after having obtained the $\Omega(k)$ lower bound result in \cref{thm:quantum_lower_bound}, used LLM assistance to potentially see if they could prove an $\Omega(k/\epsilon)$ bound, and get the right dependence on $\epsilon$. At this point, the author had an $O(k/\epsilon)$ quantum tester. With brief back and forth discussion, the LLM pointed out that instead of using naive Fourier sampling, one could use amplitude amplification to improve the dependence on $\epsilon$ to $\sqrt{\epsilon}$. 

The main use of LLMs was for the tolerant testing problem, specifically for the classical tester. Here, the author identified the main issues of having too many subspaces to test, and that potentially an extension of Alekseychuk and Konyushok's ideas could be used to compress the domain sufficiently, so that the algorithm could brute force over every $k$-dimensional subspace. At this point, the author used LLM assistance, and upon some back and forth, was able to arrive at a formulation of the compression theorem in \cref{thm:probability_compression_theorem}, as well as some form of the two lemmas used to prove it. The LLM produced a proof which, although the author believes is probably correct, the author did not fully verify due to being inscrutable. Instead, feeding in some of the prior arguments as context, specifically that used in the quantum tolerant tester, the author was able to obtain a sketch of an argument more akin to the tools already established in the paper, and adapted this to prove \cref{lem:sampling_estimate_lemma} and \cref{lem:lin_map_estimate}, and formulated \cref{thm:probability_compression_theorem} as it is displayed now, following the LLM interaction. The author takes full responsibility for the correctness. 

No LLM assistance was used in the actual writing and drafting of this paper. Any mistakes in the above proofs (and overall in the document) are the author's own.
\section{Preliminaries}
Throughout, we let a boolean function on $n$ variables be any function $f:\{0,1\}^n\rightarrow\{-1,1\}$. This is equivalent to the definition mapping into $\{0,1\}$, but since we will be dealing with Fourier transformations, this range is more convenient to work with. We will also sometimes refer to the domain as $\mathbb{F}_2^n$, and refer to the inputs interchangeably either as strings or vectors. For two strings $x,y\in\mathbb{F}_2^n$, we write $x\oplus y$ as their coordinate wise addition modulo 2, or the XOR of $x$ and $y$. Sometimes when clear, we also drop the $\oplus$ notation and refer to it as $x+y$. We let $|x|$ be the hamming weight of the string.
\subsection{Fourier Definitions}
We start by defining the Fourier character.
\begin{definition}[Fourier Character]\label{def:fouer_character}
Given a set $S\subseteq[n]$, the Fourier character corresponding to the set $S$ (or simply the Fourier character) is given by:
\begin{equation*}
    \chi_S(x)=\prod_{i\in S}(-1)^{ x_i}=(-1)^{\sum_{i\in S}x_i}.
\end{equation*}
    Abusing notation, and letting $S$ also denote the indicator vector for the set $S$, we can write the Fourier character as
    \begin{equation*}
        \chi_S(x)=(-1)^{\langle S,x\rangle},
    \end{equation*}
    where the inner product is the standard dot product.
\end{definition}
We can now define the Fourier transform.
\begin{definition}[Fourier Transform]\label{def:fourier_transform}
    The Fourier transform of a boolean function $f$ is given by
    \begin{equation*}
        f(x)=\sum_{S\subseteq[n]}\hat{f}(S)\chi_S(x),
    \end{equation*}
    where $\hat{f}(S)$ are known as the Fourier coefficients. For ease, sometimes instead of referring to $S$ as sets, we again abuse notation and let $S$ denote the indicator vector for the set $S$. They themselves are functions $\hat{f}:2^{[n]}\rightarrow\mathbb{R}$ given by
    \begin{equation*}
        \hat{f}(S)=\mathbb{E}_x[f(x)\chi_S(x)]=2^{-n}\sum_xf(x)\chi_S(x).
    \end{equation*}
\end{definition}
A standard identity about the Fourier coefficients comes from Parseval's:
\begin{theorem}[Parseval's Identity]\label{thm:parseval}
    \begin{equation*}
        \sum_{S\subseteq[n]}\hat{f}^2(S)=\mathbb{E}_x[f(x)^2]=1,
    \end{equation*}
    where we get the last equality by using the fact that $f(x)$ has range $\{\pm1\}$.
\end{theorem}
\begin{definition}[Fourier Dimension]\label{def:fourier_dim}
    A boolean function $f$ has \emph{Fourier dimension} $k$ if 
    \begin{equation*}
        \fdim(f) \coloneqq \dim(\vspan(\supp \hat{f})) = k.
    \end{equation*}
    In other words, the set of the strings corresponding to nonzero Fourier coefficients span a subspace of dimension $k$.
\end{definition}
To conclude this section, we will let $\mathcal{D}_k$ be the set of all functions with Fourier dimension at most $k$. Further, for any subspace $U$, we will let
\begin{equation*}
    C_U=\{g:\supp(\hat{g})\subseteq U\},
\end{equation*}
and let $\Pi_U$ be the projector onto $U$. Hence we have that
\begin{equation*}
    \Pi_U f(x)=\sum_{S\in U}\hat{f}(S)\chi_S(x).
\end{equation*}
\subsection{Property Testing}
We start by defining the notion of distance for boolean functions.
\begin{definition}[Distance]\label{def:distance}
    Let $f$ and $g$ be two boolean functions. The distance between $f$ and $g$ is given by 
    \begin{equation*}
        \dist(f,g)=\Pr_x[f(x)\neq g(x)].
    \end{equation*}
\end{definition}

We now outline the main problem we are concerned with.
\begin{definition}[Fourier Dimension Testing]\label{def:fourier_dim_testing}
    Let $f$ be a boolean function, which is promised to either:
    \begin{itemize}
        \item have Fourier dimension at most $k$, $\fdim(f) \leq k$ (known as the \yes-instance), or
        \item be $\epsilon$-far from any function with Fourier dimension at most $k$ (known as the \no-instance).
    \end{itemize}
    Given query access to $f$, what is the minimum number of queries required to determine which case $f$ belongs to with probability at least $2/3$?
\end{definition}
The choice of $2/3$ is arbitrary, and it can be amplified to any constant probability of success with standard arguments. For classical testers, they will get classical query access to the function, whereas for quantum testers, they will get quantum queries to the function. In other words, classical testers may query a point $x$, and in return obtain the value $f(x)$, whereas quantum testers will get a sample via an oracle 
\begin{equation*}
    O_f:|x\rangle|b\rangle\rightarrow|x\rangle|b\oplus f(x)\rangle.
\end{equation*} 
We will also consider the \emph{tolerant} version of this problem. There have been various formulations of tolerant testing, but we will use the following applied to Fourier dimensionality:
\begin{definition}[Tolerant Fourier Dimensionality Testing]\label{def:tolerant_fdim_testing}
    Let $0\leq\alpha<1/2$ and $0<\epsilon<1/2$. Given query access to a boolean function $f$, which is promised to either be:
    \begin{enumerate}
        \item At distance at most $\alpha\epsilon$ from having Fourier dimension at most $k$ (known as the \yes-instance), or
        \item At distance at least $\epsilon$ from having Fourier dimension at most $k$ (known as the \no-instance),
    \end{enumerate}
    an algorithm must accept the former, and reject the latter, with probability at least $2/3$.
\end{definition}
Again, we are interested in the query complexity of any tester. We refer to \cref{def:fourier_dim_testing} as the \emph{standard} or \emph{non-tolerant} version of the problem, and \cref{def:tolerant_fdim_testing} as the tolerant version. We comment briefly on the restriction of parameters in the above definition. The restriction on $\epsilon$ is standard, since boolean functions that are maximally far from the property correspond to $\epsilon=1/2$.\footnote{This is by observing both constant functions satisfy the Fourier dimension property, then taking a triangle inequality.} However, the restriction on $\alpha$ is a result of our algorithm and it is an open question as to whether we can get a tolerant tester for the full range of $\alpha$. Finally, when $\alpha=0$, we recover the regular testing problem defined in \cref{def:fourier_dim_testing}.

Let us conclude by defining, for any subspace $U$,
\begin{itemize}
    \item $\delta_U(f)=\dist(f,C_U)$
    \item $\wt(U)=\sum_{S\in U}\hat{f}(S)^2$
    \item $\wt(\bar{U})=\sum_{S\not\in U}\hat{f}(S)^2$.
\end{itemize}
The definition of $\wt(\cdot)$ is exactly the same as used by \cite{DBLP:journals/siamcomp/GopalanOSSW11}. Though it omits which function $f$ it refers to, for our purposes, it will always be clear from context which function we are evaluating the weights of.

To conclude, sometimes for shorthand, we drop the ``at most'', and just test against functions with Fourier dimension strictly $k$. This is because any smaller subspace can be extended to one of dimension $k$.
\subsection{Miscellaneous Group Theory}\label{ss:group}
We also briefly recall some group-theoretic notation. For a subgroup $H$, and a fixed vector $g$, we let $g+H$ be the set $\{g+h:h
\in H\}$. We let $\langle a_1,\cdots ,a_k\rangle$ denote the set of vectors generated by $a_1,\dots,a_k$. In other words, it is the set of vectors spanned by $\{a_1,\dots,a_k\}$. We call the $a_i$ the generators. A set of generators can be seen to partition the full space into cosets as follows. Let $A=\{a_i\}$ be a set of generators for a subgroup $H$. Then, $A$ partition the space $\mathbb{F}_2^n$ into $2^n/|H|$ cosets: to obtain the cosets, we sequentially choose any fixed vector $g_i$ (not already in a coset) and assign to its coset all the vectors in the set $g_i+H$, repeating until every vector in the space falls in a coset. We then call each $g_i$ a coset representative. For shorthand, we may also use $g_i$ to refer to the coset as well.
\section{Quantum Tester for Fourier Dimensionality}
In this section, we exhibit a quantum algorithm for testing Fourier dimensionality. It relies on the ability for quantum algorithms to sample according to the Fourier spectrum. Namely, we have the following:
\begin{lemma}\label{lem:fourier_sampling}
    Let $f$ be a boolean function. Then, there is an algorithm which uses $1$ quantum query and produces the state
    \begin{equation*}
        |\psi\rangle=\sum_{s\in\mathbb{F}_2^n}\hat{f}(s)|s\rangle.
    \end{equation*}
\end{lemma}
Sampling a string according to the magnitude of the Fourier coefficient simply comes from measuring the state $|\psi\rangle$. The other main lemma we will use, for which we defer the proof later, is as follows:
\begin{lemma}\label{lem:Fourier_mass}
    For any subspace $U$, we have that
    \begin{equation*}
        2\delta_U(f)\leq \wt(\bar{U}).
    \end{equation*}
\end{lemma}
As a simple corollary of the above, we now get that for any subspace $U$, $f$ will have nontrivial Fourier mass outside that subspace if it is far. Notice that we can rewrite the set $\mathcal{D}_k$ (the set of all functions with Fourier dimension at most $k$) as the union $\mathcal{D}_k=\bigcup_{\dim(U)\leq k}C_U$. This gives the following result.
\begin{corollary}\label{cor:fourier_mass}
    If $\dist(f,\mathcal{D}_k)\geq \epsilon$, then $\wt(\bar{U})\geq 2\epsilon$, for all $\dim(U)\leq k$.
\end{corollary}
We now exhibit our algorithm, and explain the implications of \cref{cor:fourier_mass}.
\begin{algorithm}[H]
    \caption{Quantum Fourier Dimensionality Tester}\label{alg:quantum_fourier_dim_tester}
    \begin{algorithmic}[1]
    \State Initialize $W=\{0\}$ and $m=O(k)$.
    \For {$i$ in $1,\dots,m$}
        \State Use \cref{lem:fourier_sampling} to produce the state $|\psi\rangle$, with a hard cap of $O(1/\sqrt{\epsilon})$ queries.. ``Mark'' all strings corresponding to $s\not\in W$, and use another $O(1/\sqrt{\epsilon})$ samples to amplitude amplify towards these marked strings.
        \State Measure this state to obtain a string $s$, and set $W=\vspan(W,s)$.
        \State Reject if $\dim(W)> k$.
    \EndFor
    \State Accept if the algorithm reaches this point.
    \end{algorithmic}
\end{algorithm}
We make two comments about the above algorithm. First, is that at a high level, this algorithm relies on \cref{cor:fourier_mass}, which states that as long as the subspace $W$ observed in the algorithm has dimension at most $k$, if $f$ is $\epsilon$-far from having Fourier dimension $k$, then the probability of Fourier sampling obtaining a string outside the subspace $W$ is at least $2\epsilon$. Thus, we can continue sampling strings, and reject if the subspace observed becomes too large. Second, though, is that we do not use \cref{lem:fourier_sampling} directly to produce a sample. This is because if we directly did this, \cref{cor:fourier_mass} implies we would need $O(1/\epsilon)$ samples to find a string which extends the dimension of the subspace $W$, which, if we repeated for $O(k)$ rounds, would give a query complexity of $O(k/\epsilon)$. However, before measuring to produce a sample, we could use amplitude amplification to increase the probability of finding a string which would extend the dimension of $W$. Since these strings initially have probability at least $\Omega(\epsilon)$, by \cref{cor:fourier_mass}, this means we need $O(1/\sqrt{\epsilon})$ samples to amplitude amplify onto these states. Once this is done, we will find a string $s$ that extends the subspace $W$, meaning this saves a factor of $\sqrt{\epsilon}$. To formalize this, we show the following:
\begin{theorem}\label{thm:quantum_tester}
    \cref{alg:quantum_fourier_dim_tester} is a one sided quantum tester for Fourier dimensionality that uses $O(k/\sqrt{\epsilon})$ quantum queries to a function $f$ and:
    \begin{itemize}
        \item Accepts with probability $1$ if $\fdim(f)\leq k$.
        \item Rejects with probability at least $2/3$ if $f$ is $\epsilon$-far from having Fourier dimension $k$.
    \end{itemize}
\end{theorem}
We will prove this, assuming the truth of \cref{lem:Fourier_mass} for now.
\begin{proof}
    We start by proving the acceptance first. Trivially, if $\fdim(f)\leq k$, the above sampling algorithm will never find a subspace $W$ of dimension greater than $k$, and thus will always reach the end of the algorithm and accept. Now, assume that $f$ is $\epsilon$-far from having Fourier dimension $k$. Then, by \cref{cor:fourier_mass}, at any point in the algorithm, there is Fourier mass of size at least $2\epsilon$ outside the current subspace $W$. Thus, starting with $|\psi\rangle$, the state where each string appears proportional to the magnitude of its Fourier coefficients, with $O(1/\sqrt{\epsilon})$ queries, amplitude amplification followed by measurement produces a string $s$ not in the span of the current $W$, thus expanding the dimension of $W$. Then, repeating this process for at least $k+1$ rounds, the algorithm will find, with high probability, a subspace $W$ of dimension $k+1$. The probability can be adjusted to be $2/3$ by adjusting the constants inside $m=O(k)$, the number of rounds, and $O(1/\sqrt{\epsilon})$, the number of samples used for amplitude amplification. Thus, if $f$ is far from having Fourier dimension $k$, the tester will find a subspace of dimension $k+1$ with probability at least $2/3$, and thus reject. This tester uses $O(k)O(1/\sqrt{\epsilon})=O(k/\sqrt{\epsilon})$ queries. 
\end{proof}
\subsection{Proof of \cref{lem:Fourier_mass}}
In this section, we conclude by proving \cref{lem:Fourier_mass}, and thus by implication \cref{cor:fourier_mass}. We start with a result by \cite{DBLP:journals/siamcomp/GopalanOSSW11}, which we include a proof of for completeness \cite[Fact 11]{DBLP:journals/siamcomp/GopalanOSSW11}.
\begin{lemma}\label{lem:perp_averaging}
    $\Pi_Uf(x)=\mathbb{E}_{h\in U^\perp}[f(x+h)]$.
\end{lemma}
\begin{proof}
By taking the Fourier transform,
    \begin{align*}
        f(x+h)&=\sum_{S\in\mathbb{F}_2^n}\hat{f}(S)\chi_S(x+h)\\\
        &=\sum_{S\in\mathbb{F}_2^n}\hat{f}(S)\chi_S(x)\chi_S(h).
    \end{align*}
    Then taking expectations, 
    \begin{equation*}
        \mathbb{E}_{h\in U^\perp}[f(x+h)]=\sum_{S\in\mathbb{F}_2^n}\hat{f}(S)\chi_S(x)\mathbb{E}_{h\in U^\perp}[\chi_S(h)].
    \end{equation*}
    We now consider some cases for $S$ and the expectation. First, if $S\in U$, then we have $\langle S, h\rangle=0$, which implies $\chi_S(h)=1$, and thus $\mathbb{E}_{h\in U^\perp}[\chi_S(h)]=1$. On the other hand if $S\not\in U$, then we can find some vector $h_0\in U^\perp$ such that $\langle S, h_0\rangle=1$. Then, we can pair up all vectors $v\in U^\perp$ by cosets $\{v, v+h_0\}$. Then, we evaluate $\chi_S(\cdot)$ on each term in this pair. The former gives $\chi_S(v)$, whereas the latter gives $\chi_S(v+h_0)=\chi_S(v)\chi_S(h_0)=-\chi_S(v)$. Hence, each pair zeroes out, meaning the expectation also evaluates to zero, $\mathbb{E}_{h\in U^\perp}[\chi_S(h)]=0$. This then gives that
    \begin{align*}
        \mathbb{E}_{h\in U^\perp}[f(x+h)]&=\sum_{S\in U}\hat{f}(S)\chi_S(x)\\
        &=\Pi_U f(x).
    \end{align*}
\end{proof}
This implies that for any subspace $U$, the value of $\Pi_U f(x)$ can be interpreted as the average value of $f$ on the coset $x+U^\perp$. With this in hand, we will now prove \cref{lem:Fourier_mass}.
\begin{proof}[Proof of \cref{lem:Fourier_mass}]
    First, note that the closest $g\in C_U$ to $f$ is the one that on every coset of $U^\perp$, is equal to the majority value of $f$. Let this function be $g_U(x)$. Then, we have that formally
    \begin{equation*}
        g_U(x)=\sgn\left(\mathbb{E}_{h\in U^\perp}[f(x+h)]\right),
    \end{equation*}
    where if the expectation is $0$, we will assign $g_U(x)=1$. Note by \cref{lem:perp_averaging}, this implies that
    \begin{equation*}
        g_U(x)=\sgn\left(\Pi_U f(x)\right).
    \end{equation*}
    This also implies that $\delta_U(f)=\dist(f,g_U)$. By construction, this also implies that on every coset of $U^\perp$, $f$ and $g_U$ disagree on a minority of the values. Assume they disagree on a $q$ fraction, where $0\leq q\leq 1/2$. We then have that, by taking an absolute value of the expectation
    \begin{align*}
        |\Pi_U f(x)|&=\left|\mathbb{E}_{h\in U^\perp}[f(x+h)]\right|\\
        &=1-2q\\
        q&=\frac{1-|\Pi_U f(x)|}{2}.
    \end{align*}
    Hence, 
    \begin{equation*}
        \delta_U(f)=\frac{1}{2}\left(1-\mathbb{E}_x[|\Pi_U f(x)|]\right).
    \end{equation*}
    On the other hand, by Parseval's theorem, we have that
    \begin{align*}
        \wt(\bar{U})&=1-\wt(U)\\
        &=1-\sum_{S\in U}\hat{f}(S)^2\\
        &=1-\sum_{S\in U} \widehat{\Pi_U f}(S)^2\\ 
        &=1-\sum_{S\subseteq[n]}\widehat{\Pi_U f}(S)^2\\ 
        &=1-\mathbb{E}_x[(\Pi_Uf(x))^2].
    \end{align*}
    The second to last line we get by recalling the projector is zero outside of $U$, and hence we can add those terms back in without changing the sum, and the last equality we get via Parseval's theorem. Recall that for any $|t|\leq 1$, we have that $1-|t|\leq 1 - t^2$. Then, since $|\Pi_U f(x)|\leq 1$, we have that
    \begin{align*}
        1-|\Pi_U f(x)|&\leq 1-|\Pi_U f(x)|^2\\
        1-\mathbb{E}_x[|\Pi_U f(x)|]&\leq 1-\mathbb{E}_x[|\Pi_U f(x)|^2]\\ 
        2\delta_U(f)&\leq \wt(\bar{U}).
    \end{align*}
\end{proof}
\subsection{An $\Omega(k)$ Lower Bound for Testing Fourier Dimensionality}\label{ss:lower_bound}
We conclude this discussion by proving an $\Omega(k)$ lower bound for testing Fourier Dimensionality, showing \cref{alg:quantum_fourier_dim_tester} is essentially tight.
\begin{theorem}\label{thm:quantum_lower_bound}
    Fix $\epsilon < 1/4$. Any algorithm which distinguishes between a boolean function $f$ having Fourier dimension at most $k$, or being $\epsilon$-far from being so requires at least $\Omega(k)$ quantum queries.
\end{theorem}
We will prove this via reduction to a lower bound in Simon's problem. We start by defining this problem.
\begin{definition}[Simon's Problem]\label{def:Simons}
    We are given (quantum) query access to a function $g:\mathbb{F}_2^m\rightarrow\mathbb{F}_2^m$ that is promised to satisfy one of the following:
    \begin{enumerate}
        \item $g$ is injective.
        \item $g$ is periodic. That is, there is some nonzero string $s$ such that $g(x)=g(x+s)$.
    \end{enumerate}
    The goal is for the quantum algorithm to distinguish between these two cases, with high probability.
\end{definition}
Koiran, Nesme and Portier gave a lower bound for this problem, which we will use as a black box.
\begin{theorem}[{\cite[Theorem 1 (Paraphrased)]{DBLP:conf/icalp/KoiranNP05}}]\label{thm:Simons_lb}
    Any algorithm that distinguishes between the two cases of Simon's problem with probability greater than $2/3$ requires at least $\Omega(m)$ queries.
\end{theorem}
They actually gave an explicit bound in terms of the lower order terms and the coefficients, but this loose $\Omega(m)$ bound suffices for our purposes. The second lemma we will need is that if $f\in C_U$ (that is, the support of the strings with nonzero Fourier coefficients of $f$ lies in $U$), then $f$ is constant on all cosets of $U^\perp$.
\begin{lemma}\label{lem:constant_cosets}
     $f\in C_U$, if and only if for any $h\in U^\perp$ and $x\in\mathbb{F}_2^n$,
    \begin{equation*}
        f(x+h)=f(x).
    \end{equation*}
\end{lemma}
\begin{proof}
To show the forward direction, we explicitly write out the Fourier transform:
    \begin{align*}
        f(x)&=\sum_{S\in\mathbb{F}_2^n}\hat{f}(S)\chi_S(x)\\
        &=\sum_{S\in U}\hat{f}(S)\chi_S(x).
    \end{align*}
    Now, for any $h\in U^\perp$, a similar calculation to the above will show $\chi_S(x+h)=\chi_S(x)$, giving the result as required.

Next, to show the reverse direction, by \cref{lem:perp_averaging}, if $f(x+h)=f(x)$ for all $h\in U^\perp$, then we immediately get
\begin{equation*}
    \Pi_Uf(x)=f(x),
\end{equation*}
which implies that $f\in C_U$, as required.
\end{proof}
The intuition for our proof is as follows: Notice the second case, where $g$ is periodic, essentially coincides with $g$ being constant over a subspace of dimension $1$ (that is, the dimension of the subspace spanned by $s$ alone). In our prior notation, this implies $U^\perp=\vspan(s)$, and thus $g$ has Fourier dimension at most $m-1$, by \cref{lem:constant_cosets}. Then, if we had a Fourier dimension tester as a black box subroutine (ignoring for now, the dependence on the distance parameter $\epsilon$), we could just test what the Fourier dimension of $g$ is, and if it is at most $m-1$, we know it is periodic, otherwise it is injective. This works as an intuition, but to make this work, we need to change our function into the form of Simon's problem. Namely, the function in Simon's problem maps strings to $\mathbb{F}_2^m$, but our Fourier dimension tester only works on boolean functions. To make this work, our Fourier dimension tester can just sample a uniformly random function $M:\mathbb{F}_2^m\rightarrow\{\pm 1\}$, and apply that to $g$. In other words, we will perform the following:
\begin{enumerate}
    \item Take $g$ as in \cref{def:Simons}.
    \item Sample a uniformly random function $M:\mathbb{F}_2^m\rightarrow\{\pm1\}$. We denote this $M$ since it serves to ``mask'' the original Simon's function.
    \item Compose the two to create the function $f(x)=M(g(x))$. Notice now that $f:\mathbb{F}_2^m\rightarrow \{\pm1\}$. 
\end{enumerate}
Now, we can test the Fourier dimension of $f$. If g is periodic, then clearly $f$ is also periodic, since
\begin{equation*}
    f(x)=M(g(x))=M(g(x+s))=f(x+s),
\end{equation*}
and thus for the same reason, $\fdim(f)\leq m-1$. On the other hand, assume $g$ is injective. Then, $f$ is just a random function from $\mathbb{F}_2^m\rightarrow\{\pm1\}$. We can show that with high probability, a random function is far from having Fourier dimension $m-1$.
\begin{lemma}\label{lem:random_fourier_dimension}
    Fix the distance parameter $\epsilon < 1/4$. Then, with probability $1-o(1)$ over the choice of $M$, we have that $f$ is $\epsilon$-far from having Fourier dimension $k$, and the $o(1)$ parameter depends on $\epsilon$, with $g$ being an injectiv function.
\end{lemma}
\begin{proof}
    Fix some nonzero period $s$, which decomposes the domain into $2^{m-1}$ pairs, $\{x,x+s\}$. To make $f$ periodic, we require that each pair evaluates to the same bit. For a uniformly random $f$, each pair is equal with probability $1/2$. Let $B_s$ be the number of differing pairs for the fixed $s$. This is distributed binomially, since disjoint pairs are independent, and hence we have $B_s\sim\bin(2^{m-1},1/2)$. For each of these differing pairs, one value needs to be altered so that $f(x)=f(x+s)$, which implies that
    \begin{equation*}
        \dist(f,\{\text{functions with period $s$}\})=\frac{B_s}{2^m}.
    \end{equation*}
    By a Hoeffding bound, using that $\mathbb{E}[B_s]=2^{m-2}$, we get that
    \begin{align*}
        \Pr[B_s\leq\epsilon 2^m]&\leq \exp\left(-\frac{(1-4\epsilon)^2}{2}2^{m-2}\right)\\
        \Pr\left[\frac{B_s}{2^m}\leq\epsilon\right]&\leq \exp(-c_\epsilon 2^m).
    \end{align*}
    Now, taking a union bound over all possible choices of $s$:
    \begin{equation*}
        \Pr[f\text{ is }\epsilon-\text{close to some periodic function}]\leq 2^m\exp(-c_\epsilon 2^m)=o(1).
    \end{equation*}
    Hence, with probability $1-o(1)$, a random function is $\epsilon$-far from a periodic function, meaning with probability $1-o(1)$ over the choice of $M$, $f$ is $\epsilon$-far from having Fourier dimension at most $m-1$.
\end{proof}
With this setup, we can now prove \cref{thm:quantum_lower_bound}.
\begin{proof}[Proof of \cref{thm:quantum_lower_bound}]
    We prove this bound via a reduction to Simon's problem, and then leveraging the lower bound of Koiran, Nesme and Portier in \cref{thm:Simons_lb}. Assume we are given a function $g$, which satisfies \cref{def:Simons}, and we are given a black box algorithm $\mathcal{A}$, a quantum tester for Fourier dimension. Set $m=k+1$, and if needed, one can pad to $n$ variables by adding redundant variables (the tester can even be explicitly told which indices here are the redundant variables). First, uniformly at random sample a mask function $M:\mathbb{F}_2^m\rightarrow\{\pm 1\}$, and create the function $f(x)=M(g(x))$. If $g$ is periodic, $f$ is also periodic, and hence $f$ has Fourier dimension at most $k$. Otherwise, by \cref{lem:random_fourier_dimension}, it is $\epsilon$-far from having Fourier dimension $k$ with high probability. A phase query to $f$ can be simulated using $O(1)$ queries to $g$, first by querying $g$ in superposition, then applying the transformation $M$ to the resulting state in superposition, and then uncomputing $g$, if needed. Now, we can solve the Simon's decision problem as follows. Run $\mathcal{A}$ on $f$, and return periodic whenever $\mathcal{A}$ determines $f$ has Fourier dimension at most $k$, otherwise return injective. Repeat for a constant amount of repetitions, over different choices of mask $M$, and take a majority vote. This suffices as an algorithm for Simon's problem. By \cref{thm:Simons_lb}, then $\mathcal{A}$ requires at least $\Omega(k)$ queries.
\end{proof}
\section{An (Almost) Tight Classical Tester for Fourier Dimensionality}
The purpose of this section is to give a classical tester for Fourier dimensionality that is tight with respect to the lower bound of \cite{DBLP:journals/siamcomp/GopalanOSSW11}, up to lower order factors. We present the algorithm below.
\begin{algorithm}[H]
    \caption{A Classical Tester for Fourier Dimensionality}\label{alg:classical_fourier_dim_tester}
    \begin{algorithmic}[1]
    \State Set $T=O(1)$, $m=O(k/\epsilon)$, $\ell=O(2^{k/2})$, $\lambda={\ell\choose 2}/2^k$ and $\tau=(1-e^{-\lambda}+0.51\lambda)/2$.
    \For {$t$ in $1,\dots,T$}
        \State Uniformly at random sample $m$ shifts $h_1,\dots,h_m$ and $\ell$ points $x_1,\dots,x_\ell$.
        \State For each $x_i$, construct the tuple $\Phi(x_i)=(f(x_i+h_1),\dots f(x_i+h_m))$.
        \State Record ``A'' if there exists $i\neq j$ such that $\Phi(x_i)=\Phi(x_j)$.
    \EndFor
    \State Accept if the number of recorded ``A'' is at least $\tau\cdot T$.
    \end{algorithmic}
\end{algorithm}
To explain the intuition behind this algorithm a bit more, it helps first to consider an ``idealized'' algorithm. Assume, in this ideal world, $f$ has Fourier dimension (at most) $k$. Then, the subspace spanned by its strings with nonzero Fourier coefficients has dimension (at most) $k$, which induces (at most) $2^k$ cosets. However, if $f$ is far from having Fourier dimension at most $k$, then its strings with nonzero Fourier coefficients span a subspace of at least dimension $k+1$, which implies that it has at least $2^{k+1}$ cosets. Assume now, upon each query $x$, the tester does not get $f(x)$, but instead learns which coset $x$ lies in (with some numbering or ordering on the cosets such that this makes sense, and that somehow the convenient numbering does not give away which instance we are in). Then, with $\ell=\Theta(2^{k/2})$ uniform samples, we could distinguish the \yes- and \no- instances simply by observing the number of collisions, which are more abundant in the former than in the latter. 

Of course though, we are not in this ``ideal'' world. Queries do not inform the algorithm of the coset, and it just gives a bit, which not granular enough to distinguish cosets. It would be nice though, if we could still somehow, with high probability, implement this ``coset notification'' idea in proxy. To do this, the starting point is the following consequence of \cref{lem:constant_cosets}.
\begin{corollary}\label{cor:constant_tuples}
    Let $x,y$ belong to the same coset defined by the function $f$. Then, $\Phi(x)=\Phi(y)$, for all randomly sampled $h_1,\dots,h_m$.
\end{corollary}
This is simply a consequence of strings in the same coset evaluating to the same output. Then the idea is that instead of comparing say $f(x_i)$ and $f(x_j)$ for a collision, we instead look at the tuples $\Phi(x_i)$ and $\Phi(x_j)$ for queries, which are a lot more granular. These then, can serve as proxies for this ``coset notification'' idea we had in the ideal world. This would work if we see more collisions in $\Phi(\cdot)$ in the \yes-instances as compared to the \no-instances. Fortunately, we do, and we quantify this in the following two lemmas.  
\begin{lemma}\label{lem:yes_instance_collision}
    Assume $\fdim{f}\leq k$. Then, in one execution of Lines 3-5 in \cref{alg:classical_fourier_dim_tester}, the probability of recording ``A'' is at least
    \begin{equation*}
        \Pr[\text{Recording A}]\geq 1-e^{-\lambda}.
    \end{equation*}
\end{lemma}
\begin{lemma}\label{lem:no_instance_collision}
    Assume $f$ is $\epsilon$-far from having Fourier dimension $k$. Then, in one execution of Lines 3-5 in \cref{alg:classical_fourier_dim_tester}, the probability of recording ``A'' is at most
    \begin{equation*}
        \Pr[\text{Recording A}]\leq 0.51\lambda.
    \end{equation*}
\end{lemma}
Assuming the truth of the above two lemmas for now, we aim to prove the following result:
\begin{theorem}\label{thm:classical_tester}
    Let $\epsilon>0$ be fixed. Then \cref{alg:classical_fourier_dim_tester} uses $O\left(\frac{2^{k/2}k}{\epsilon}\right)=\tilde{O}(2^{k/2}/\epsilon)$ queries, and distinguishes between $f$ having Fourier dimension at most $k$, and $f$ being $\epsilon$-far from having Fourier dimension $k$ with probability at least $2/3$.
\end{theorem}
\begin{proof}
    By adjusting the constant inside $\ell$, we can choose $\lambda$ such that $\lambda\in[1/2,1]$. Then, notice in this case,
    \begin{equation*}
        1-e^{-\lambda}-0.51\lambda\geq 0.1,
    \end{equation*}
    that is, there is a constant sized gap. Choose $\tau=\frac{ 1-e^{-\lambda}+0.51\lambda}{2}$. Now, assume we repeat $T$ repetitions of Lines 3-5, and round $i$ outputs $X_i=1$ if that round records A, and otherwise outputs $X_i=0$. Let
    \begin{equation*}
        \bar{X}=\frac{1}{T}\sum_{i=1}^TX_i,~~\mu=\mathbb{E}[\bar{X}].
    \end{equation*}
    In other words, our algorithm will accept if $\bar{X}\geq \tau$. By Hoeffding's inequality, we have that
    \begin{equation*}
        \Pr[|\bar{X}-\mu|\geq \gamma]\leq 2\exp(-2T\gamma^2).
    \end{equation*}
    Then, the probability the algorithm errors is at most the probability that $|\bar{X}-\mu|\geq \Delta:=\frac{1-e^{-\lambda}-0.51\lambda}{2}$, which is a constant. Plugging this into Hoeffding's inequality, we get that
    \begin{align*}
        \Pr[|\bar{X}-\mu|\geq \Delta]&\leq 2\exp(-2T\Delta^2)\\
        &\leq 2\exp\left(-\frac{T}{2}(1-e^{-\lambda}-0.51\lambda)^2\right)\\
        &\leq \frac{1}{3},
    \end{align*}
    where the last inequality is obtained by adjusting the constant inside $T$. Thus, the algorithm errors with probability at most $1/3$. The final thing to count is the number of samples used: The algorithm uses $m\ell$ samples in each round, for a total of $T$ rounds. Hence, the total number of samples to $f$ given by
    \begin{equation*}
        Tm\ell=O(1)O(k/\epsilon)O(2^{k/2})=O\left(\frac{k2^{k/2}}{\epsilon}\right),
    \end{equation*}
    as required.
\end{proof}
\subsection{Proof of Two Lemmas}
We conclude with the proofs of \cref{lem:yes_instance_collision} and \cref{lem:no_instance_collision}. We start with the simpler one, the former. This one relies only on a simple balls and bins argument.
\begin{proof}[Proof of \cref{lem:yes_instance_collision}]
    Assume $\fdim(f)\leq k$, and $U=\vspan(\supp\hat{f})$. This implies $\dim(U)\leq k$. By \cref{cor:constant_tuples}, if two points $x_i$ and $x_j$ are in the same coset, then $f(x_i+h)=f(x_j+h)$ for all $h$, including all the sampled shifts. In other words, $\Phi(x_i)=\Phi(x_j)$. By assumption, there are at most $2^k$ cosets. For simplicity, assume there are exactly $2^k$ cosets, and $\Phi(x)$ and $\Phi(y)$ are different if $x$ and $y$ are in different cosets. This assumption will only decrease the probability of collision, and thus not recording ``A''. We then get that 
    \begin{align*}
        \Pr[\text{Not Recording A}]&\leq \prod_{i=0}^{\ell-1}\left(1-\frac{i}{2^k}\right)\\
        &\leq \exp\left(-\frac{{\ell\choose 2}}{2^k}\right)\\
        &=e^{-\lambda}.
    \end{align*}
    Therefore, the probability of seeing a collision, and thus recording ``A'', is at least $1-e^{-\lambda}$.
\end{proof}
We now assume $f$ is far from having Fourier dimension $k$, meaning $\fdim(f)\geq k+1$, and aim to show that collisions are less abundant in this regime. The main intermediate point we want to build up to is the following lemma, which we will use in the proof of \cref{lem:no_instance_collision}. Interestingly, this relates the probability of the tuples being equal to a probability related to Fourier sampling.
\begin{lemma}\label{lem:fixed_tuple_fourier_prob}
    Let $x$ and $y$ be fixed, and $h=x+y$. Define the Fourier distribution $\mu_f$, where $\mu_f(a)=\hat{f}(a)^2$. Then,
    \begin{equation*}
        \Pr_{h_1,\dots,h_m}[\Phi(x)=\Phi(y)]=\Pr_{S_1,\dots,S_m\sim\mu_f}[\langle S_i,h\rangle=0,~\text{for all }i\in[m]],
    \end{equation*}
    where $h_i$ are the randomly sampled shifts defined in \cref{alg:classical_fourier_dim_tester}.
\end{lemma}
In essence this allows us to relate the probability that $\Phi(x)=\Phi(y)$ for fixed $x$ and $y$ to a sampling problem in the Fourier space. Namely, the probability of collision is equivalent to the probability that after $m$ samples from the Fourier spectrum of $f$, we get a subspace that is entirely orthogonal to $x+y$. To build up to this, we first prove a smaller result for a singular shift.
\begin{lemma}\label{lem:fixed_fourier_collision}
    Let $h=x+y$, for $x,y\in\mathbb{F}_2^n$ fixed. Then, for uniformly sampled $z\sim\mathbb{F}_2^n$,
    \begin{equation*}
        \Pr_z[f(x+z)=f(y+z)]=\sum_{S:\langle S,h\rangle=0}\hat{f}(S)^2.
    \end{equation*}
\end{lemma}
\begin{proof}
    Consider $u=x+z$. Since $z$ is uniformly sampled, so is $u$. Then, notice that $y+z=x+z+h=u+h$, recalling that $h=x+y$. Thus
    \begin{equation*}
        \Pr_z[f(x+z)=f(y+z)]=\Pr_u[f(u)=f(u+h)].
    \end{equation*}
    Using that $f\in\{\pm1\}$, we get that
    \begin{equation*}
        \mathbbm{1}[f(u)=f(u+h)]=\frac{1+f(u)f(u+h)}{2}.
    \end{equation*}
    Noting that this right hand side is bounded between zero and one, we get that 
    \begin{align*}
        \Pr_u[f(u)=f(u+h)]&=\mathbb{E}\left[\frac{1+f(u)f(u+h)}{2}\right]\\
        &=\frac{1}{2}(1+\mathbb{E}_u[f(u)f(u+h)]).
    \end{align*}
    Taking Fourier expansions of both $f(u)$ and $f(u+h)$, we get that 
    \begin{equation*}
        \mathbb{E}_u[f(u)f(u+h)]=\sum_{S}\hat{f}(S)^2\chi_S(h),
    \end{equation*}
    which means that 
    \begin{align*}
        \Pr_u[f(u)=f(u+h)]&=\frac{1}{2}(1+\sum_S\hat{f}(S)^2(-1)^{\langle S, h\rangle})\\
        &=\frac{1}{2}\sum_S\hat{f}(S)^2(1+(-1)^{\langle S, h\rangle}),
    \end{align*}
    where in the second line we used \cref{thm:parseval}. Now notice that $1+(-1)^{\langle S,h\rangle}=0$ if $\langle S,h\rangle=1$. Thus, it is only nonzero when $\langle S,h\rangle=0$, in which case it equals two, proving the result.
\end{proof}
Using this, we can now prove \cref{lem:fixed_tuple_fourier_prob}.
\begin{proof}[Proof of \cref{lem:fixed_tuple_fourier_prob}]
    Using \cref{lem:fixed_fourier_collision}, for any particular $h_i$, we have that 
    \begin{equation*}
        \Pr_{h_i}[f(x+h_i)=f(y+h_i)]=\sum_{\langle S,h\rangle=0}\hat{f}(S)^2=\Pr_{S\sim\mu_f}[\langle S,h\rangle=0].
    \end{equation*}
    Since all $h_i$'s are sampled independently, we have that, by definition of the tuples being equal,
    \begin{align*}
        \Pr_{h_1,\dots, h_m}[\Phi(x)=\Phi(y)]&=\prod_{i=1}^m\Pr_{S_i\sim \mu_f}[\langle S_i,h\rangle=0]\\
        &=\Pr_{S_1,\dots,S_m\sim\mu_f}[\langle S_i,h\rangle=0,~\text{for all }i\in[m]].
    \end{align*}
\end{proof}
We then drop the assumption on $x$ and $y$ being fixed, and instead randomly sampled, as needed in \cref{alg:classical_fourier_dim_tester}.
\begin{corollary}\label{cor:tuple_fourier_collisions}
 Sample $x,y,h_1,\dots,h_m$ uniformly at random.
    \begin{equation*}
        \Pr_{x,y,h_1,\dots,h_m}[\Phi(x)=\Phi(y)]=\mathbb{E}_{S_1,\dots,S_m\sim\mu_f}[2^{-\dim(\vspan(S_1,\dots,S_m))}].
    \end{equation*}
\end{corollary}
\begin{proof}
    Let $K=\vspan(S_1,\dots,S_m)$. Since $x$ and $y$ are u.a.r, so is $h=x+y$. Then, $\langle S_i,h\rangle=0$ for all $i$ if and only if $h\in K^\perp$. Assume $\dim(K)=d$, which implies $\dim(K^\perp)=n-d$, and $|K^\perp|=2^{n-d}$. Then,
    \begin{equation*}
        \Pr_h[h\in K^\perp]=\frac{2^{n-d}}{2^n}=2^{-d}=2^{-\dim(K)}.
    \end{equation*}
    Now, we take the expectation over $S_i$'s sampled from $\mu_f$, which gives
    \begin{equation*}
        \Pr_{x,y,h_1,\dots,h_m}[\Phi(x)=\Phi(y)]=\mathbb{E}_{S_1,\dots,S_m\sim\mu_f}[2^{-\dim(\vspan(S_1,\dots,S_m))}].
    \end{equation*}
\end{proof}
This is all we need to prove \cref{lem:no_instance_collision}.
\begin{proof}[Proof of \cref{lem:no_instance_collision}]
    Consider any pair $x_i$ and $x_j$. By \cref{cor:tuple_fourier_collisions}, analyzing the probability $\Pr[\Phi(x_i)=\Phi(x_j)]$ is equivalent to analyzing the expected dimension of subspace spanned by $m$ Fourier samples according to $f$. Let this sequence of Fourier samples be $S_1,\dots,S_m$. We consider the size of the dimension of the spans as we process the Fourier samples. Let $K_t=\vspan(S_1,\dots,S_t)$, and suppose $\dim(K_t)\leq k$. Then, by \cref{cor:fourier_mass}, 
    \begin{equation*}
        \Pr_{S_{t+1}}[S_{t+1}\not\in K_t]\geq 2\epsilon,
    \end{equation*}
    and if this occurs, this implies $\dim(K_{t+1})=1+\dim(K_t)$. Hence, until we have a sample $t$ such that $\dim(K_t)=k+1$, every new sample increases the dimension of the Fourier subspace with probability at least $2\epsilon$. We model this with a binomial distribution. Namely, let $B\sim\bin(m,2\epsilon)$, and let $d=\dim(\vspan(S_1,\dots,S_m))$. Note if $B\leq k$, then $B\leq d$ implies that $2^{-d}\leq 2^{-B}$. Else, if $B\geq k+1$, then $d\geq k+1$  and thus $2^{-d}\leq 2^{-k-1}$. In other words, we can conclude that 
    \begin{equation*}
        2^{-d}\leq 2^{-k-1}+2^{-B}.
    \end{equation*}
    Hence, by \cref{cor:tuple_fourier_collisions},
    \begin{align*}
        \Pr[\Phi(x)=\Phi(y)]&=\mathbb{E}[2^{-d}]\\
        &\leq 2^{-k-1}+\mathbb{E}[2^{-B}]\\
        &\leq 2^{-k-1}+e^{-\epsilon m}.
    \end{align*}
    The last inequality we get because $\mathbb{E}[2^{-B}]=(1-\epsilon)^m$, and then applying the inequality on this term. Choosing $m=O(k/\epsilon)$, and the constant appropriately, we get that 
    \begin{equation*}
        \Pr[\Phi(x)=\Phi(y)]\leq \frac{0.51}{2^k}.
    \end{equation*}
    Then, taking a union bound over all $\ell\choose 2$ pairs, we get that
    \begin{equation*}
        \Pr[\text{at least one collision}]\leq {\ell\choose 2}\frac{0.51}{2^k}=0.51\lambda.
    \end{equation*}
\end{proof}
\section{A Tolerant Quantum Tester}\label{sec:tolerant_quantum_tester}
We move on now from the task of (standard) Fourier dimensionality testing, and instead consider tolerant testing. Notice that a property tester does not immediately imply a tolerant tester. This is because the tester in \cref{alg:quantum_fourier_dim_tester} does not have guarantees for any function that is between $0$ and $\epsilon$-far from having Fourier dimension at most $k$. In fact, one can show it will actually reject functions in this range with probability $\Omega(1)$, which does not suffice for the tolerant testing problem. However, with a natural extension of \cref{alg:quantum_fourier_dim_tester}, we show there is an algorithm with comparable query complexity that does solve the tolerant testing problem. Namely, we show the following:
\begin{theorem}\label{thm:tolerant_quantum_ub}
    For fixed $\alpha$ and $\epsilon$ satisfying the conditions in \cref{def:tolerant_fdim_testing}, there is a tolerant tester for Fourier dimensionality that uses
    \begin{equation*}
        O\left(\frac{k}{\epsilon(1-2\alpha)^2}\log\frac{1}{\epsilon(1-2\alpha)}\right)
    \end{equation*}
    quantum queries.
\end{theorem}
The denominator here shows why we need the $\alpha<1/2$ assumption. Before presenting the algorithm, we give a high level overview of the idea behind it. In the tolerant setting, instead of testing whether some function's Fourier spectrum is entirely contained in a $k$-dimensional subspace, as we do in \cref{alg:quantum_fourier_dim_tester}, we instead want to count how many such samples fall outside of a subspace of dimension $k$. Ideally, this number should be higher in a \no-instance as compared to a \yes-instance. Via \cref{lem:Fourier_mass}, we already know in a \no-instance, we have at least $2\epsilon$ Fourier mass outside of any $k$-dimensional subspace. What we are missing is the equivalent for the \yes-instance. We give this now:
\begin{lemma}\label{lem:yes_mass_prob}
    Assume $\dist(f,g)\leq \alpha\epsilon$, where $g$ has Fourier dimension at most $k$, and moreover let $\vspan(\supp\hat{g})=U$. Then, for a string $s$ sampled from the Fourier spectrum of $f$,
    \begin{equation*}
        \Pr[s\not\in U]\leq 4\alpha\epsilon.
    \end{equation*}
\end{lemma}
\begin{proof}
    The proof is a simple combination of Parseval's, and the fact that $\hat{g}(s)=0$ whenever $s\not\in U$.
    \begin{align*}
        \Pr[s\not\in U]&=\sum_{s\not\in U}\hat{f}(s)^2\\
        &=\sum_{s\not\in U}[\hat{f}(s)-\hat{g}(s)]^2\\
        &\leq \sum_s[\hat{f}(s)-\hat{g}(s)]^2\\
        &=\mathbb{E}_x[(f(x)-g(x))^2]\\
        &=4\dist(f,g)\\
        &\leq 4\alpha\epsilon.
    \end{align*}
\end{proof}
Hence in light of this lemma, the approach is clear. We will repeatedly take Fourier samples and compare the number of samples that fall outside a subspace of dimension (at most) $k$. Unlike in the non-tolerant version, we will not be maintaining a subspace, and merely counting the number of strings that fall outside. We present the formal algorithm now:
\begin{algorithm}[H]
    \caption{Tolerant Quantum Fourier Dimensionality Tester}\label{alg:tolerant_quantum_fourier_dim_tester}
    \begin{algorithmic}[1]
        \State Set $m=O\left(\frac{k}{\epsilon(1-2\alpha)^2}\log\frac{1}{\epsilon(1-2\alpha)}\right)$ and $\tau=(1+2\alpha)\epsilon m$.
        \State Perform $m$ Fourier samples on $f$.
        \State For every subspace $U$ with dimension at most $k$, count how many of the Fourier sampled strings lie outside $U$.
        \State Accept if for some subspace $U$, the number of strings lying outside is at most $\tau$. Else, if no such subspace satisfies this, reject.
    \end{algorithmic}
\end{algorithm}
We now prove \cref{alg:tolerant_quantum_fourier_dim_tester} solves the tolerant Fourier dimensionality testing problem, and that it has the claimed query complexity, hence proving \cref{thm:tolerant_quantum_ub}.
\begin{proof}[Proof of \cref{thm:tolerant_quantum_ub}]
    We handle the query complexity first. Conditional on the algorithm actually solving the tolerant testing problem, it clearly only uses 
    \begin{equation*}
        m=O\left(\frac{k}{\epsilon(1-2\alpha)^2}\log\frac{1}{\epsilon(1-2\alpha)}\right)
    \end{equation*}
    quantum queries.

    \paragraph{Yes instance. }Start by considering \yes-instances. In this case, there is some function $g$ that has Fourier dimension $k$, and the distance from $f$ to $g$ is minimal and bounded by $\alpha\epsilon$. Let $U=\vspan(\supp\hat{g})$. Let $X$ be the random variable denoting the number of Fourier samples which land outside $U$. By \cref{lem:yes_mass_prob}, this is a binomial random variable, with mean $\mathbb{E}[X]\leq 4\alpha\epsilon m$. We can also bound, with a similar calculation, that $\var[X]\leq 4\alpha\epsilon m$, using the variance of a binomial random variable. Then, using Chebyshev's we get that
    \begin{align*}
        \Pr[\text{Algorithm rejects}]&\leq \Pr[X>\tau]\\
        &\leq \Pr\left[\left|X-\mathbb{E}[X]\right|\geq (1-2\alpha)\epsilon m\right]\\
        &\leq \frac{4\alpha\epsilon m}{\epsilon^2 m^2(1-2\alpha)^2}\\
        &\leq \frac{4\alpha}{\epsilon m(1-2\alpha)^2},
    \end{align*}
    which is at most $1/3$, by changing the constant inside $m$. To justify the first line, if $X\leq \tau$, the algorithm always accepts. However, if $X>\tau$, the algorithm does not necessarily reject, since some other subspace may have fewer than $\tau$ samples outside it, in which case the algorithm will accept regardless. We also note for completeness, that there is no need for the extra $\log$ factor in $m$ for this analysis, that will come for the \no-instance.

    \paragraph{No instance. }Next we move onto the \no-instances. The trickier part with this analysis was that in the \yes-instance, we fixed some reference subspace independent of the actual Fourier samples. However, in the \no-instances, there is no reference subspace, since the subspaces that may cause issues will actually depend on the Fourier samples. To handle this case, we consider the following argument instead. Let the Fourier samples be denoted $s_1,\dots,s_m$. Suppose the tester erroneously accepts in this case. Then, we can find some subspace $U$, of dimension at most $k$, such that it has at most $\tau$ samples outside of its span. Then, we can define the subspace
    \begin{equation*}
        V=\vspan\{s_i:s_i\in U\}.
    \end{equation*}
    By construction, any sample outside of $U$ is also outside of $V$, and any sample inside of $U$ also lies in $V$. Also letting $j=\dim (V)$, by construction $k\geq j$. Reframing the analysis in terms of $V$, this now implies that if the tester accepts, there is some choice of (at most) $k$ basis vectors in the samples, such that there are at most $\tau$ samples which lie outside of the span of these vectors. For now, fix a particular $V$, and choose $j$ samples to serve as its basis. Let $Y$ denote the random variable indicating the number of samples that lie outside of $V$. By \cref{lem:Fourier_mass}, this has mean at least $\mathbb{E}[Y]\geq 2\epsilon(m-j)$. Then, the probability that the tester erroneously accepts on this particular fixed $V$ is bounded by
    \begin{equation*}
        \Pr[\text{tester erroneously accepts on $V$}]\leq \Pr[Y\leq \tau].
    \end{equation*}
    We will bound this via a (multiplicative) Chernoff bound. For this, we require that
    \begin{align*}
        (1-\gamma)\mathbb{E}[Y]&\geq \tau\\
        2\epsilon(1-\gamma)(m-j)&\geq (1+2\alpha)\epsilon m\\
        2(1-\gamma)m&\geq (1+2\alpha)m\\
        1-\gamma&\geq\frac{1+2\alpha}{2}\\
        \gamma&\leq \frac{1-2\alpha}{2}.
    \end{align*}
    Hence, choose $\gamma=\frac{1-2\alpha}{2}$, which holds for $m\geq\frac{3+2\alpha}{1-2\alpha}k$. Then, applying the Chernoff bound:
    \begin{align*}
        \Pr[Y\leq \tau]&\leq \Pr[Y\leq (1-\gamma)\mathbb{E}[Y]]\\
        &\leq \exp\left(-\frac{\gamma^2\mathbb{E}[Y]}{2}\right)\\
        &\leq \exp(-\Omega(\epsilon m(1-2\alpha)^2)).
    \end{align*}
    To conclude, we note this was for the error given a particular fixed $V$. We need to take a union bound over all such possible $k$ dimensional $V$ induced by the samples, of which there are at most
    \begin{equation*}
        {m\choose k}\leq \left(\frac{em}{k}\right)^k.
    \end{equation*}
    Thus, we have that 
    \begin{equation*}
        \Pr[\text{algorithm errors}]\leq \left(\frac{em}{k}\right)^k\exp(-\Omega(\epsilon m(1-2\alpha)^2)),
    \end{equation*}
    which is bounded by $1/3$ by adjusting the constant inside $m$.
\end{proof}
We conclude with a few comparisons to \cref{thm:quantum_tester}. First, our tolerant tester includes an extra $\log$ factor. While this arises because of the Chernoff bound in the \no-instances, it is unclear whether this is an artifact of the proof method, or whether it is actually necessary. Secondly, notice that when $\alpha=0$, we should expect to recover exactly the bound in \cref{thm:quantum_tester}. However, disregarding even the $\log$ factor, we do not get the square root saving in the $\epsilon$ factor. This is because in the standard version, we were able to maintain a subspace, and amplitude amplify to find strings outside of the current subspace. However, in this case, it is unclear what the ``marked'' states should be at every step in this algorithm, and thus we cannot use amplitude amplification to obtain this saving. We leave it open as to whether this argument can be refined to allow for such a saving.
\section{A Tolerant Classical Tester}\label{sec:classical_tolerant}
We next move on to the classical tolerant tester. Before explaining our algorithm for this problem, it helps to study an easier problem first. As an intuitive starting point, suppose someone informed the algorithm of the function $g$, the closest function to $f$ with Fourier dimension at most $k$, and let $U=\vspan(\supp\hat{g})$. Then in this case, courtesy of \cref{lem:fixed_fourier_collision}, we have a fairly intuitive tester. Consider instead
\begin{equation*}
    \Pr[f(x)\neq f(x+h)]=\sum_{\langle s,h\rangle=1}\hat{f}(s)^2,
\end{equation*}
the complement event of \cref{lem:fixed_fourier_collision}. By the same arguments above, we have that for any $s\in U$, we have that $\langle s,h\rangle=0$, whereas for $s\not\in U$,
\begin{equation*}
    \Pr_{h\in U^\perp}[\langle s,h\rangle=1]=\frac{1}{2}.
\end{equation*}
Taking expectations, we get that 
\begin{equation*}
    \mathbb{E}_{h\in U^\perp}\left[\Pr_x[f(x)\neq f(x+h)]\right]=\frac{1}{2}\sum_{s\not\in U}\hat{f}(s)^2.
\end{equation*}
Now, in a \yes-instance, this latter quantity is at most $2\alpha\epsilon$, by \cref{lem:yes_mass_prob}, whereas in a \no-instance, it is at least $\epsilon$, by \cref{cor:fourier_mass}. Since we know $U$, we could now just sample $h\in U^\perp$, query $f(x)$ and $f(x+h)$, and count the number of disagreements, which will be higher in a \no-instance, and thus conclude which case $f$ belongs to. However, the issue with making this intuition concrete is that we do not know the right $U$ to use. We could try to naively test them all, as we do in \cref{alg:tolerant_quantum_fourier_dim_tester}. The issue is that in \cref{alg:tolerant_quantum_fourier_dim_tester}, the samples are taken independently of the subspaces. Here, we specifically need to sample from $U^\perp$. We can save some queries with smart recycling of queries between different subspaces, but the core issue is that there are just too many candidate subspaces to test (roughly of order ${2^n\choose k}$, which is far too large).

\subsection{A Compression Theorem}
Our approach to address the above is inspired by that of Alekseychuk and Konyushok. To test whether a function $f$ has Fourier dimension $k$, they prove, at a high level, that after applying a random linear map $X:\mathbb{F}_2^n\rightarrow\mathbb{F}_2^r$ (with $r$ of sufficient dimension), then $f$ has Fourier dimension $k$ if and only if the compressed function $f_X=f(X^\top u)$ also has Fourier dimension $k$. The benefit of this is that $f_X$ has $r$ variables, meaning it has fewer subspaces that you need to enumerate. Our approach can be thought of as a tolerant version of this approach. We will aim to show that after applying a random linear map and compressing the function, the minimum Fourier mass outside every $k$-dimensional subspace remains roughly similar. To make this formal, we first need to introduce some notation to formalize what we are doing:
\begin{definition}[Probability Notation]\label{def:prob_notations}
    Throughout, let $f$ be a boolean function. We will let
    \begin{enumerate}
        \item $\mu$ be an arbitrary probability distribution on $\mathbb{F}_2^n$. In particular, we denote $\mu_f$ as the Fourier distribution, where $\mu_f(s)=\hat{f}(s)^2$. When clear, we omit the subscript $f$.
        \item For a linear map $X$, let $\mu_X(t)=\sum_{s:Xs=t}\mu(s)$. In particular for the Fourier distribution, let $\mu_{f,X}=\sum_{s:Xs=t}\hat{f}(s)^2$. Again, when clear, we will omit the subscript $f$, though for clarity will always include the linear map $X$.
        \item Let $\rho_k(u)$ be the minimum probability mass outside of any subspace of dimension $k$. That is,
        \begin{equation*}
            \rho_k(\mu)=\min_{\dim W\leq k}\mu(\bar{W}).
        \end{equation*}
        When a linear map is applied, we also consider the analogous $\rho_{k,X}(\mu)$,
        \begin{equation*}
            \rho_{k,X}(\mu)=\min_{\dim W\leq k}\mu_X(\bar{W}).
        \end{equation*}
        In particular, when considering the Fourier distribution, let $\rho_{k,X}(f)$ be the minimum Fourier mass outside of any $k$-dimensional subspace. That is, 
        \begin{equation*}
            \rho_{k,X}(f)=\min_{\dim W\leq k}\mu_{f,X}(\bar{W}).
        \end{equation*}
    \end{enumerate}
\end{definition}
Phrased in the language of this definition, we aim to show that under a random linear transform (of sufficient dimension), $\rho_{k,I}(f)$ is roughly the same as $\rho_{k,X}(f)$. We now present the main theorem. While it is phrased for general probability distributions, and thus may be of independent interest, the reader who prefers a concrete example can always think of $\mu$ as corresponding to the Fourier distribution induced by their favorite boolean function.
\begin{theorem}\label{thm:probability_compression_theorem}
    Let $\mu$ be a probability distribution on $\mathbb{F}_2^n$, and fix $0<\beta<1$. Then, over a uniformly random linear map $X:\mathbb{F}_2^n\rightarrow\mathbb{F}_2^r$, where $r=O(k\log\frac{1}{\beta})$, with probability at least $0.99$, we have that 
    \begin{equation*}
        \rho_{k}(\mu)-\beta\leq \rho_{k,X}(\mu)\leq \rho_k(\mu).
    \end{equation*}
\end{theorem}
Notice that in the compressed space, we have roughly $2^{kr}$ $k$ dimensional subspaces, a large number still, but now independent of $n$. To prove this theorem, we will need two intermediate results. The first is a sampling algorithm that estimates $\rho_k(\mu)$. The point is that in proving the lower bound of \cref{thm:probability_compression_theorem}, we will not actually deal with them directly, we will prove them via the following estimates below, and show that the estimates cannot be too off from the true quantity. 
\begin{lemma}\label{lem:sampling_estimate_lemma}
    Fix $0<\beta<1$, $\delta>0$, and set $M=O\left(\frac{k\log(1/\beta)+\log(1/\delta)}{\beta^2}\right)$. Consider a sequence of $M$ samples from $\mu$, $s_1,\dots,s_M$, and define the estimate
    \begin{equation*}
        \tilde{\rho}_k=\frac{1}{M}\min_{\dim W\leq k}|\{i:s_i\not\in W\}|.
    \end{equation*}
    Then,
    \begin{equation*}
        \Pr\left[\left|\tilde{\rho}_k-\rho_k(\mu)\right|\leq\frac{\beta}{2}\right]\geq 1-\delta.
    \end{equation*}
\end{lemma}
There are two things worth pointing out at this point. First, is that nowhere in this lemma do we mention the linear map $X$, and second is that the choice of $M$ might seem odd. Indeed, if the goal was to get a good estimate, one should choose $M$ as comically large as possible. Both of these points are addressed in the following lemma. Namely, in the next, we will show that the random linear map $X$ does not change the estimate at all, with high probability, and also that there is a connection between $M$ and $r$, so if $M$ is too large, $r$ will also end up being large, and thus the compression in \cref{thm:probability_compression_theorem} still results in too many subspaces to effectively enumerate. Of course, one cannot then take $M$ too small, to be very economical in the number of subspaces you need to enumerate, since then the estimate is too poor. Hence, this value of $M$ is in some sense living in the ``sweet spot'' for our purposes. Now that we have addressed this, the second result we need is as follows.
\begin{lemma}\label{lem:lin_map_estimate}
    Fix a set of $M$ samples from $\mu$, with $M$ in \cref{lem:sampling_estimate_lemma}, a $\delta>0$ (not necessarily the same as in \cref{lem:sampling_estimate_lemma}) and a random linear map $X:\mathbb{F}_2^n\rightarrow \mathbb{F}_2^r$, where $r=O(k\log \frac{M}{k}+\log \frac{1}{\delta})$. Then, with probability at least $1-\delta$,
    \begin{equation*}
        \tilde{\rho}_k=\tilde{\rho}_{k,X},
    \end{equation*}
    where $\tilde{\rho}_{k,X}$ is the estimator defined in the same way, except on samples $Xs_i$ instead of $s_i$.
\end{lemma}
Assuming the truth of these two lemmas, we now prove \cref{thm:probability_compression_theorem}.
\begin{proof}[Proof of \cref{thm:probability_compression_theorem}]
    We start with the upper bound, which is the easier direction. Indeed, take any subspace $W$, of dimension at most $k$, and let $U=X(W)$. Since $X$ is linear, this has dimension at most $\dim W$. We also have that 
    \begin{equation*}
        \mu_X(U)=\mu(X^{-1}(U))\geq \mu(W),
    \end{equation*}
    since by construction $W\subseteq X^{-1}(U)$. By taking complements, observing that by definition, $\rho_{k,X}(\mu)\leq \mu_X(\bar{U})$, and choosing $W$ such that $\mu(\bar{W})=\rho_k(\mu)$, we have that 
    \begin{equation*}
        \rho_{k,X}(\mu)\leq \mu_X(\bar{U})\leq \mu(\bar{W})=\rho_k(\mu).
    \end{equation*}
    Moving to the lower bound, first, draw $M$ samples from $\mu$, as required in \cref{lem:sampling_estimate_lemma}. Then, applying that lemma, we have that
    \begin{equation*}
        \left|\tilde{\rho}_k-\rho_k(\mu)\right|\leq\frac{\beta}{2}.
    \end{equation*}
    By \cref{lem:lin_map_estimate}, we also have that
    \begin{equation*}
        \tilde{\rho}_{k}=\tilde{\rho}_{k,X},
    \end{equation*}
    and applying \cref{lem:sampling_estimate_lemma} again, on the probability distribution $X\mu=\mu_X$, we have that 
    \begin{equation*}
        \left|\tilde{\rho}_{k,X}-\rho_{k,X}(\mu)\right|\leq\frac{\beta}{2}.
    \end{equation*}
    All three of these can hold with arbitrarily high probability, and thus, by taking $\delta$ small enough in each case and union bounding over the three events, all three bounds succeed with probability at least $0.99$. Then, conditional on the three holding, we can apply all three in succession to get that
    \begin{align*}
        {\rho}_{k,X}&\geq \tilde{\rho}_{k,X}-\frac\beta2\\
        &=\tilde{\rho}_{k}-\beta\\
        &\geq {\rho}_k(\mu)-\beta.
    \end{align*}
\end{proof}
The proofs of the remaining two lemmas are not too interesting, and we will defer them to \cref{app:deferred_prob_proofs}.
\subsection{Reframing Tolerant Fourier Dimensionality Testing}
By \cref{thm:probability_compression_theorem}, we can now re-characterize tolerant Fourier dimensionality testing, which we will use for our classical tester. First, fix $\beta=\epsilon(1-2\alpha)$. Then, in a \yes-instance, we have that for every possible linear map $X$,
\begin{equation*}
    \rho_{k,X}(f)\leq \rho_k(\mu_f)\leq 4\alpha\epsilon,
\end{equation*}
for some subspace $W$ of dimension at most $k$. On the other hand, for a \no-instance, with probability at least $0.99$ over choice of $X$,
\begin{align*}
    \rho_{k,X}&\geq \rho_k(\mu_f)-\epsilon(1-2\alpha)\\
    &\geq 2\epsilon-\epsilon(1-2\alpha)\\
    &=(1+2\alpha)\epsilon,
\end{align*}
for any subspace $W$ of dimension at most $k$. From here onwards, we will ignore the tiny probability of error over uniformly sampled $X$, since this can be amplified to any arbitrarily small constant with usual techniques and then union bounded away. That being said, we will concern ourselves with the following re-characterization of Fourier dimensionality testing:
\begin{definition}[Reframed Tolerant Fourier Dimensionality Testing]\label{def:reframed_tolerant_fdim_test}
    Given a boolean function $f$, and a random linear map $X$ as quantified in \cref{thm:probability_compression_theorem}, we are promised that either
    \begin{enumerate}
        \item There exists a subspace $W$ of $\mathbb{F}_2^r$ of dimension $k$, such that $\mu_{f,X}(\bar{W})\leq 4\alpha\epsilon$, or
        \item For any subspace $W$ of dimension $k$, $\mu_{f,X}(\bar{W})\geq (1+2\alpha)\epsilon$.
    \end{enumerate}
\end{definition}
We are interested in bounding the number of queries to $f$ it takes to solve this problem with probability at least $2/3$, which will give a bound on our original tolerant testing problem. It is worth re-emphasizing that even though we have compressed the space of our problem into $\mathbb{F}_2^r$, we are still only allowed queries to the original function $f$ taking inputs in $\mathbb{F}_2^n$.

\paragraph{A high level ``warmup''. } Consider the following ``warmup'' algorithm, which is essentially just using the fact we established above, though adapted to this compressed space. Fix $W$, a subspace of $\mathbb{F}_2^r$ of dimension $k$. Then, uniformly sample a point $d\in W^\perp$, and a point $x\in\mathbb{F}_2^n$. We can then query the two points $f(x)$ and $f(x+X^\top d)$, and test the event
\begin{equation*}
    \mathbbm{1}[f(x)\neq f(x+X^\top d)].
\end{equation*}
Then, taking the expectation of this event occurring over $d\in W^\perp$, we again recover something similar to \cref{lem:fixed_fourier_collision}:
\begin{equation*}
    \mathbb{E}_{d\in W^\perp}[\mathbbm{1}[f(x)\neq f(x+X^\top d)]]=\Pr_{x\in\mathbb{F}_2^n,~d\in W^\perp}[f(x)\neq f(x+X^\top d)]=\frac{1}{2}\sum_{s:Xs\not\in W}\hat{f}(s)^2=\frac{1}{2}\mu_{f,X}(\bar{W}).
\end{equation*}
We omit the proof of this identity since it is essentially the same as that of \cref{lem:fixed_fourier_collision}. In this case, there is an extra factor of $1/2$, since \cref{lem:fixed_fourier_collision} only averages over a fixed shift, whereas this averages over $d\in W^\perp$ as well. Nevertheless, this does give us a way to form a tester. We can repeatedly query points and estimate the above quantity for a fixed $W$. One can show that to test this to arbitrary accuracy, one will need to take
\begin{equation*}
    O\left(\frac{1}{\epsilon(1-2\alpha)^2}\right)\approx O(1)
\end{equation*}
queries, where for intuition, we will pretend $\epsilon$ and $\alpha$ are constants. This is for one fixed $W$, but yet again, we need to do this for all possible subspaces. What we have gained from this compression, however, is that now, we have at most $2^{O(kr)}$ subspaces, which is now independent of $n$. This would, brushing aside the intricacies, lead to a tolerant tester with roughly 
\begin{equation*}
    2^{O(kr)}=2^{O(k^2)}
\end{equation*}
queries. The issue with this tester is that it is not as sample efficient as it could be. Namely, we are possibly using fresh queries, when we could have been reusing previous queries to test for the same subspace. To remedy this in the next section, we will give our tolerant tester, which allows us to reuse queries, inspired by the approach in \cref{alg:classical_fourier_dim_tester}, to improve the query complexity. 
\subsection{Query Efficient Tolerant Classical Tester}
\begin{algorithm}[t]
    \caption{A Classical Tolerant Tester for Fourier Dimensionality}\label{alg:classical_tolerant_fourier_dim_tester}
    \begin{algorithmic}[1]
    \State Set $r=O\left(k\log\frac{1}{\epsilon(1-2\alpha)}\right)$, $\ell=O(2^{k/2})$, $T=O\left(\frac{kr}{\epsilon(1-2\alpha)^2}\right)$, $\tau=\frac{1+6\alpha}{2}\epsilon$ and $p=1-\prod_{i=0}^{\ell-1}\left(1-\frac{i}{2^k}\right)$.
    \State For every $k$ dimensional subspace $W$ of $\mathbb{F}_2^r$, choose a basis $\{w_1,\dots,w_k\}$ (this choice can be arbitrary).
    \State Sample a uniformly random linear map $X:\mathbb{F}_2^n\rightarrow\mathbb{F}_2^r$.
    \For{ $t=1,\dots,T$}
        \State Sample $x_t\in\mathbb{F}_2^n$ uniformly at random.
        \State Sample $a_{t,1},\dots,a_{t,\ell}\in\mathbb{F}_2^r$ uniformly at random.
        \State For each sample, compute for every subspace $W$ the tuple
        \begin{equation*}
            \Phi_W(a_{t,i})=(\langle w_1,a_{t,i}\rangle,\dots,\langle w_k, a_{t,i}\rangle).
        \end{equation*}
        \State Query $f(x_t+X^\top a_{t,i})$ for all $i$.
        \State For every subspace $W$, define the quantity $Z_{t,W}$, where 
        \begin{equation*}
            Z_{t,W}=\begin{cases}
                0,~~\text{if there is no $i<j$ such that $\Phi_{W}(a_{t,i})=\Phi_W(a_{t,j})$}\\
                \mathbbm{1}[f(x_t+X^\top a_{t,i})\neq f(x_t+ X^\top a_{t,j})],~~\text{otherwise.}
            \end{cases}
        \end{equation*}
        \indent If there are multiple such pairs in the second case, use an unbiased rule to pick one (e.g., the lexicographically first such pair). Let $C_{t,W}$ denote the event that such a pair exists.
    \EndFor
    \State For every $k$-dimensional subspace $W$ of $\mathbb{F}_2^r$, compute the estimate
    \begin{equation*}
        \hat{\mu}_X(W)=\frac{2}{pT}\sum_{t=1}^TZ_{t,W}.
    \end{equation*}
    \State Accept if some $W$ satisfies $\hat{\mu}_X(W)\leq \tau$, otherwise reject.
    \end{algorithmic}
\end{algorithm}
We give the classical tester in \cref{alg:classical_tolerant_fourier_dim_tester}. Notice that although the tuple in line 7 looks similar to that defined in \cref{alg:classical_fourier_dim_tester}, and even the analysis of why this line works looks very similar, it is worth explicitly outlining and contrasting why we are looking at this tuple in the two algorithms, because they serve different roles. In \cref{alg:classical_fourier_dim_tester}, we essentially used the tuples to obtain the collision rate, which then allowed us to estimate the number of cosets, and thus the Fourier dimension. In this case though, if two tuples collide, that informs us that $a_{t,i}+a_{t,j}\in W^\perp$.\footnote{This is because if this occurs, then $\langle w,a_{t,i}\rangle=\langle w,a_{t,j}\rangle$ for every basis vector $w$, which then implies $\langle w, a_{t,i}+a_{t,j}\rangle=0$.} Then, this means we can test whether $f(x_t+X^\top a_{t,i})\neq f(x_t+X^\top a_{t,j})$ and, as we will explain below, estimating the probability of this occurring will give an estimate of $\mu_{f,X}(\bar{W})$. Notice here $p$ is just the probability of seeing at least one collision in $\ell$ queries tossed uniformly into $2^k$ buckets. Before proving the theorem, we will aim to show \cref{lem:one_loop_bound}, which we will use in the theorem proof.
\begin{theorem}\label{thm:classical_tolerant_ub}
    Fix $0<\epsilon<1$ and $0\leq\alpha<1/2$. Then \cref{alg:classical_tolerant_fourier_dim_tester} distinguishes whether a function $f$ is at most distance $
    \alpha\epsilon$ from having Fourier dimension at most $k$, or is $\epsilon$-far from being so, with query complexity
    \begin{equation*}
        O\left(\frac{2^{k/2}k^2}{\epsilon(1-2\alpha)^2}\log\frac{1}{\epsilon(1-2\alpha)}\right).
    \end{equation*}
\end{theorem}
To do this we will establish an intermediate result, about one loop of the tester in \cref{thm:classical_tolerant_ub}.
\begin{lemma}\label{lem:one_loop_bound}
    Let $C_W$ be the event that there exists at least one pair $i,j$ for subspace $W$ such that $\Phi_W(a_{t,i})=\Phi_W(a_{t,j})$. Then, for every $k$-dimensional subspace $W$ of $\mathbb{F}_2^r$,
    \begin{equation*}
        \mathbb{E}[Z_W]=\frac{\Pr[C_W]}{2}\mu_{f,X}(\bar{W}).
    \end{equation*}
\end{lemma}
Here, we have omitted the subscript $t$ denoting the round, since the probability is the same every single round.
\begin{proof}
    Conditional on $C_W$ occurring, there exists a colliding pair $i$ and $j$, so we can set
    \begin{equation*}
        d=a_i+a_j\in W^\perp,
    \end{equation*}
    and $d$ is also uniform in $W^\perp$. Also, note the vector
    \begin{equation*}
        u=x+X^\top a_i
    \end{equation*}
    is also uniform in $\mathbb{F}_2^n$. Hence,
    \begin{align*}
        \mathbb{E}[Z_W|C_W]&=\mathbb{E}_{d\in W^\perp}\left[\Pr_{u\in\mathbb{F}_2^n}[f(u)\neq f(u+X^\top d)]\right]\\
        &=\frac{1}{2}\mu_{f,X}(\bar{W}).
    \end{align*}
    Then, noting that $Z_W=0$ whenever $C_W$ does not occur and undoing the conditioning on $C_W$ gives the result.
\end{proof}
It also helps to remember that $p=\Pr[C_W]\geq 1-e^{-\lambda}$, with $\lambda={\ell\choose 2}/2^k$ set as in \cref{alg:classical_fourier_dim_tester}, for the same birthday argument as presented there. The analysis immediately translates once it is observed that $a_i+a_j\in W^\perp$ only if they are both in the same coset of $W^\perp$, and that we have $2^{r-(r-k)}=2^k$ cosets of $W^\perp$ in this case too. We now give the proof for \cref{thm:classical_tolerant_ub}.
\begin{proof}[Proof of \cref{thm:classical_tolerant_ub}]
    We will again break the analysis into the two cases.

    \paragraph{Yes instances. }We start with the \yes-instances. Assume $f$ is at most $\alpha\epsilon$ far from having Fourier dimension at most $k$. Then, by \cref{lem:yes_mass_prob} and \cref{thm:probability_compression_theorem}, there exists some $k$-dimensional subspace $W$ which satisfies
    \begin{equation*}
        \mu_{f,X}(\bar{W})\leq 4\alpha\epsilon.
    \end{equation*}
    Then, by \cref{lem:one_loop_bound}, for every $t$, we have that 
    \begin{equation*}
        \mathbb{E}[Z_{t,W}]\leq 2\Pr[C_W]\alpha\epsilon.
    \end{equation*}
    This implies that the expectation of the estimate is bounded by
    \begin{equation*}
        \mathbb{E}[\hat{\mu}_X(W)]\leq 4\alpha\epsilon.
    \end{equation*}
    Furthermore, since each indicator is Bernoulli, we also have that 
    \begin{align*}
        \var[Z_{t,W}]&\leq \mathbb{E}[Z_{t,W}]\\
        \var[\hat{\mu}_X(W)]&\leq \frac{8\alpha\epsilon}{pT}\leq\frac{4\epsilon}{pT},
    \end{align*}
    using that $\alpha<1/2$. We can then apply Chebyshev's inequality and obtain that 
    \begin{align*}
        \Pr[\hat{\mu}_X(W)>\tau]&\leq \Pr\left[\left|\hat{\mu}_X(W)-\mathbb{E}[\hat{\mu}_X(W)]\right|\geq \frac{1-2\alpha}{2}\epsilon\right]\\
        &\leq \frac{16\epsilon/(pT)}{(1-2\alpha)^2\epsilon^2}\\
        &=\frac{16}{pT\epsilon(1-2\alpha)^2}.
    \end{align*}
    Hence, choosing $T=O\left(\frac{1}{\epsilon(1-2\alpha)^2}\right)$, and adjusting the constant inside suffices to make this probability an arbitrarily small constant.

    \paragraph{No instances. }Next, assume $f$ is $\epsilon$-far, then by \cref{thm:probability_compression_theorem}, over the random choice of $X$, we have that 
    \begin{equation*}
        \rho_{k,X}(f)\geq (1+2\alpha)\epsilon.
    \end{equation*}
    Then, every $k$-dimensional subspace $W$ of $\mathbb{F}_2^r$ satisfies
    \begin{equation*}
        \mu_{f,X}(\bar{W})\geq (1+2\alpha)\epsilon.
    \end{equation*}
    Combining this with \cref{lem:one_loop_bound}, we have that the expectation of the estimate for one $W$ satisfies
    \begin{equation*}
        \mathbb{E}[\hat{\mu}_X(W)]\geq (1+2\alpha)\epsilon.
    \end{equation*}
    We will use the multiplicative Chernoff bound to bound the probability the estimate is less than $\tau$. For this, we will choose $\gamma$ such that
    \begin{align*}
        \frac{1+6\alpha}{2}\epsilon&\leq (1-\gamma)(1+2\alpha)\epsilon\\
        1-\gamma&\geq\frac{1+6\alpha}{2+4\alpha}\\
        \gamma&=\frac{1-2\alpha}{2+4\alpha},
    \end{align*}
    where in the last line we switch to equality since we just want one $\gamma$ such that this holds. Applying the Chernoff bound:
    \begin{align*}
        \Pr[\hat{\mu}_X(W)\leq\tau]&\leq \Pr[\hat{\mu}_X(W)\leq (1-\gamma)\mathbb{E}[\hat{\mu}_X(W)]]\\
        &\leq\exp\left(-\frac{\gamma^2pT\mathbb{E}[\hat{\mu}_X(W)]}{4}\right)\\
        &\leq \exp\left(-\frac{pT\epsilon(1-2\alpha)^2}{32}\right).
    \end{align*}
    Having done this for one $W$, we then want to union bound over all such possible $W$. As a loose bound, there are $2^{kr}$ of these, since we need $k$-basis vectors, each of which have at most $2^r$ options to choose from. Thus, the probability that the tester incorrectly accepts is bounded by
    \begin{equation*}
        \Pr[\text{incorrect acceptance}]\leq 2^{kr}\exp\left(-\frac{pT\epsilon(1-2\alpha)^2}{32}\right).
    \end{equation*}
    It suffices to choose $T=O\left(\frac{kr}{\epsilon(1-2\alpha)^2}\right)$, and adjusting the constant insides makes this probability an arbitrarily small constant.

    \paragraph{Queries. }Having proven this tester has arbitrarily small error on both the \yes- and \no-instances, we now just need to bound the number of queries. In one round, the algorithm makes $\ell=O(2^{k/2})$ queries. This is repeated for $T=O\left(\frac{kr}{\epsilon(1-2\alpha)^2}\right)$ rounds. Setting $r=O\left(k\log\frac{1}{\epsilon(1-2\alpha)}\right)$, the algorithm in total makes
    \begin{equation*}
        T\ell=O\left(\frac{2^{k/2}k^2}{\epsilon(1-2\alpha)^2}\log\frac{1}{\epsilon(1-2\alpha)}\right)
    \end{equation*}
    queries, as claimed.
\end{proof}
\printbibliography
\appendix
\section{Proof of Two Lemmas}\label{app:deferred_prob_proofs}
In this section, we give the two proofs that were deferred for the lemmas used for \cref{thm:probability_compression_theorem}. We start with \cref{lem:sampling_estimate_lemma}.
\begin{proof}[Proof of \cref{lem:sampling_estimate_lemma}]
    We will separately prove that both 
    \begin{align*}
        \tilde{\rho}_k&<\rho_k(\mu)-\frac\beta2\\
        \tilde{\rho_k}&>\rho_k(\mu)+\frac\beta2
    \end{align*}
    happen with small probability. For the latter, we can bound it with Hoeffding's. Fix a subspace $U$ which attains $\rho_k(\mu)$. Then, the mean of $\tilde{\rho}_k$ at least the mean of $M$ independent Bernoulli trials, each with mean $\rho_k(\mu)$. By Hoeffdings, this then gives
    \begin{equation*}
        \Pr\left[\tilde{\rho}_k>\rho_k(\mu)-\frac\beta2\right]\leq\exp(-2M(\beta/2)^2),
    \end{equation*}
    which is at most $\delta/2$ by taking $M$ as stated, and possibly adjusting the constant.

    For the former bound, we will use a method similar to the soundness in \cref{thm:tolerant_quantum_ub}. Suppose that $\tilde{\rho_k}<\rho_k(\mu)-\frac\beta2$. Then, there exists some subspace $U$ that contains at most $M(\rho_k(\mu)-\frac\beta2)$ samples. We replace $U$ by the span $V$ of the sampled points that lie in $U$. In other words, we can find a basis of $j\leq k$ sampled points such that 
    \begin{equation*}
        V=\vspan\{s_i:s_i\in U\}.
    \end{equation*}
    These $j$ points are now inside of $V$, whilst the remaining $M-j$ points are still distributed via $\mu$. Let $Y$ be the random variable denoting how many of the remaining samples land in $V$. Then we can interpret $Y$ as a binomial random variable, and get
    \begin{equation*}
        \mathbb{E}[Y]=(M-j)\mu (\bar{V})\geq (M-j)\rho_k(\mu).
    \end{equation*}
    For this case to hold, we need 
    \begin{align*}
        \mathbb{E}[Y]-Y&>(M-j)\rho_{k}(\mu)-M\left(\rho_k(\mu)-\frac\beta2\right)\\
        &=\frac{\beta M}{2}-j\rho_{k}(\mu)\\
        &\geq\frac{\beta M}{4},
    \end{align*}
    provided $M$ is large enough. Plugging this into another Hoeffding bound, and enforcing that $M\geq 4k/\beta$, we get that
    \begin{align*}
        \Pr\left[|Y-\mathbb{E}[Y]|\geq\frac{\beta M}{4}\right]&\leq2\exp\left(-\frac{2(\beta M/4)^2}{M-j}\right)\\
        &\leq2\exp(-\beta^2M/8).
    \end{align*}
    This was for one subspace, and then we can take a union bound over all possible dimension $k$ subspaces induced by the queries, of which there are at most
    \begin{equation*}
        {M\choose k}\leq \left(\frac{eM}{k}\right)^k.
    \end{equation*}
    This gives that 
    \begin{equation*}
        \Pr\left[\tilde{\rho}_k<\rho_k(\mu)-\frac\beta2\right]\leq 2\left(\frac{eM}{k}\right)^k\exp(-\beta^2M/8),
    \end{equation*}
    which is at most $\delta/2$ by taking $M$ as stated, and adjusting the constant if needed.

    A union bound over these two cases completes the proof.
\end{proof}
We now prove the second.
\begin{proof}[Proof of \cref{lem:lin_map_estimate}]
    We start by observing that a linear map cannot increase the rank, and thus the only case where the estimator can error is if some subspace of higher dimension is compressed down into dimension $k$. Every such subspace of rank greater than $k$ contains $k+1$ linearly independent vectors in the $M$ samples. Fix one such $(k+1)$-tuple. These will become linearly dependent under $X$ only if some nonzero vector in this span lies in $\ker X$. For each nonzero vector $v$, since the image is uniform in $\mathbb{F}_2^r$, the probability $v$ gets mapped to the zero vector is 
    \begin{equation*}
        \Pr_X[Xv=0]=2^{-r}.
    \end{equation*}
    Taking a union bound over the $2^{k+1}-1$ nonzero vectors in the span of this $(k+1)$-tuple gives
    \begin{equation*}
        \Pr_X[\text{This $(k+1)$-tuple collapses}]\leq 2^{k+1-r}.
    \end{equation*}
    Next, we also union bound over the at most ${M\choose k+1}$ different $(k+1)$-tuples, which gives that 
    \begin{equation*}
        \Pr_X[\text{some tuple collapses}]\leq {M\choose k+1}2^{k+1-r}\leq \left(\frac{eM}{k+1}\right)^{k+1}2^{k+1-r}.
    \end{equation*}
    This is at most $\delta$ by taking $r$ as stated, with possibly adjusting the constant inside.
\end{proof}
\end{document}